\documentclass[11pt,a4paper]{article}
\usepackage[utf8]{inputenc}

\usepackage[utf8]{inputenc} 
\usepackage[T1]{fontenc}    
\usepackage{hyperref}       
\usepackage{url}            
\usepackage{booktabs}       
\usepackage{amsfonts}       
\usepackage{nicefrac}       
\usepackage{microtype}      
\usepackage{xcolor}         
\usepackage{times}

\usepackage{soul}
\usepackage{url}
\usepackage[small]{caption}
\usepackage{graphicx}
\usepackage{amsmath}
\usepackage{booktabs}
\usepackage{natbib}
\usepackage{fullpage}
\usepackage{algorithmic}
\usepackage{epsfig}
\usepackage{epstopdf}
\usepackage{graphicx}
\usepackage{latexsym}
\usepackage{amsmath}
\usepackage{amsfonts}
\usepackage{amssymb}

\usepackage{mathrsfs}
\usepackage{float}
\usepackage{pifont}
\usepackage{yhmath}
\usepackage{bbm}
\usepackage{eufrak}
\usepackage{amsthm}
\usepackage{booktabs}
\usepackage{verbatim}
\usepackage{wasysym}
\usepackage{array}
\usepackage[ruled]{algorithm2e}
\usepackage{caption} 
\usepackage{graphicx}
\usepackage{subfigure} 
 \usepackage{makecell}
\usepackage{amsmath,amssymb,amsfonts,graphics}
\newtheorem{theorem}{Theorem}[section]
\newtheorem{lemma}[theorem]{Lemma}

\newtheorem{claim}[theorem]{Claim}

\newtheorem{definition}[theorem]{Definition}

\newtheorem{observation}[theorem]{Observation}

\theoremstyle{remark}

\newcommand{\bR}{{\mathbb{R}}}

\newcommand{\MMS}{\mathsf{MMS}}
\newcommand{\sw}{\mathsf{sw}}

\newcommand{\bX}{\mathbf{X}}
\newcommand{\cI}{\mathcal{I}}
\newcommand{\cG}{\mathcal{G}}
\newcommand{\cE}{\mathcal{E}}
\newcommand{\w}{\mathbf{w}}

\title{Fair Graphical Resource Allocation with \\ Matching-Induced Utilities}

\author{Zheng Chen$^1$ \hspace{30pt} Bo Li$^2$ \hspace{30pt} Mingming Li$^3$  \hspace{30pt} Guochuan Zhang$^1$\\
$^1$College of Computer Science and Technology, Zhejiang University, Hangzhou, China\\
\texttt{\{21721122, zgc\}@zju.edu.cn}\\
$^2$ Department of Computing, The Hong Kong Polytechnic University, Hong Kong\\
\texttt{comp-bo.li@polyu.edu.hk}\\
$^3$ Department of Computer Science, City University of Hong Kong, Hong Kong\\
\texttt{minming.li@cityu.edu.hk}
}

\date{April 2022}
\begin{document}
\maketitle

\sloppy

\begin{abstract}
Motivated by real-world applications, we study the fair allocation of graphical resources, where 
the resources are the vertices in a graph. 
Upon receiving a set of resources, an agent's utility equals the weight of a maximum matching in the induced subgraph. 
We care about maximin share (MMS) fairness and envy-freeness up to one item (EF1).
Regarding MMS fairness, the problem does not admit a finite approximation ratio for heterogeneous agents.
For homogeneous agents, we design constant-approximation polynomial-time algorithms, and also note that significant amount of social welfare is sacrificed inevitably in order to ensure (approximate) MMS fairness.
We then consider EF1 allocations whose existence is guaranteed.  
However, the social welfare guarantee of EF1 allocations cannot be better than $1/n$ for the general case, where $n$ is the number of agents.
Fortunately, for three special cases, binary-weight, two-agents and homogeneous-agents, we are able to design polynomial-time algorithms that also ensure a constant fractions of the maximum social welfare.
\end{abstract}
\section{Introduction}
Resource allocation has been actively studied due to its practical applications \citep{DBLP:books/daglib/0017734,DBLP:journals/sigecom/GoldmanP14,flanigan2021fair}. 
Traditionally, the utilities are assumed to be additive which means an agent's value for a bundle of resources equals the sum of each single item's marginal utility. 
But in many real-world problems, the resources have graph structures and thus the agents' utilities are not additive but depend on the structural properties of the received resources. 
For example, Peer Instruction (PI) has been shown to be an effective learning approach based on a project conducted at Harvard University, and one of the simplest ways to implement PI is to pair the students \citep{crouch2001peer}.
Consider the situation when we partition students to advisors, where the advisors will adopt PI for their assigned students. 
Note that the advisors may hold different perspectives on how to pair the students based on their own experience and expertise, and they want to maximize the efficiency of conducting PI in their own assigned students. 
How should we assign the students fairly to the advisors?
How can we maximize the social welfare among all (approximately) fair assignments?
In this work, we take an algorithm design perspective to solve these two questions.
Similar pairwise joint work also appears as long-trip coach driver vs. co-driver and accountant vs. cashier, which is widely investigated in matching theory \citep{lovasz2009matching}. 

The graphical nature of resources has been considered in the literature (see, e.g., \citep{DBLP:conf/ijcai/BouveretCEIP17,DBLP:journals/dam/Suksompong19,DBLP:conf/innovations/BiloCFIMPVZ19,DBLP:conf/aaai/IgarashiP19}).
In this line of research, the graph is used to characterize feasible allocations, such as the resources allocated to each agent should be connected, but the agents still have additive utilities over allocated items. 
We refer the readers to \citep{DBLP:journals/sigecom/Suksompong21} for a comprehensive survey of constrained fair division. 
As shown by the previous PI an other examples, with graphical resources, the value of a set of resources does not solely depend on the vertices or the edge weights, but decided by the combinatorial structure of the subgraph, such as the maximum matching in our problem.

Our problem also aligns the research of balanced graph partition \citep{DBLP:journals/eor/MiyazawaMOW21}.  
Although there are heuristic algorithms in the literature \citep{DBLP:journals/eor/KressMP15,DBLP:journals/dam/BarketauPS15} that partition a graph when the subgraphs are evaluated by maximum matchings,
these algorithms do not have theoretical guarantees.
Our first fairness criterion is the {\em maximin share} (MMS) fairness proposed by \cite{budish2011combinatorial}, which generalizes the max-min objective in Santa Claus problem \citep{DBLP:conf/stoc/BansalS06}.
Informally, the MMS value of an agent is her best guarantee if she is to partition the graph into several subgraphs but receives the worst one.
We aim at designing efficient algorithms with provable approximation guarantees. 
As will be clear later, to achieve (approximate) MMS fairness, a significant amount of social welfare has to be inevitably sacrificed. 
Our second fairness notion is {\em envy-freeness} (EF) \citep{foley1967resource}.
In an EF allocation, no agent prefers the allocation of another agent to her own.
Since the resources are indivisible, such an allocation barely exists, and recent research in fair division focuses on achieving its relaxations instead.
One of the most widely accepted and studied relaxations is {\em envy-freeness up to one item} (EF1) \citep{budish2011combinatorial}, which requires the envy to be eliminated after removing one item. 
\cite{DBLP:conf/sigecom/LiptonMMS04} proved that an EF1 allocation always exists even with combinatorial valuations.\footnote{The algorithm in \citep{DBLP:conf/sigecom/LiptonMMS04}  was originally published in 2004 with a different targeting property. In 2011, \citet{budish2011combinatorial} formally proposed the notion of EF1 fairness.}
It is noted that an arbitrary EF1 allocation may have low social welfare, and our goal is to compute an EF1 allocation which preserves a large fraction of the maximum social welfare without fairness constraints.
The social welfare loss by enforcing the allocations to be EF1 is quantified by {\em price of EF1} \citep{DBLP:journals/mst/BeiLMS21}.

\subsection{A Summary of Results}

We study the fair allocation of graphical resources when the resources are indivisible and correspond to the {\em vertices} in the graph, and the agents' valuations are measured by the weight of the maximum matchings in the induced subgraphs.
The fairness of an allocation is measured by maximin share (MMS) and envy-free up to one item (EF1).
Our model strictly generalizes the additive setting: it degenerates to the additive setting when the graph consists of a set of independent edges by regarding each edge as an item whose value is the weight of the edge.
This is because the removal of a vertex also removes the adjacent edge.
We aim at designing efficient algorithms that
compute fair allocations with high social welfare.
Our main results are summarized as follows.

We first consider MMS fairness and find that no algorithm has bounded approximation ratio even if there are two agents with binary weights. 
We thus focus on the homogeneous case when the agents have identical valuations. 
Then our problem degenerates to the max-min objective, i.e., partitioning the vertices so that the minimum weight of the maximum matchings in the subgraphs is maximized. 
It is easy to see that an MMS fair allocation always exists but finding it is NP-hard. 
Accordingly, we design a polynomial-time $1/8$-approximation algorithm for arbitrary number of agents, and show that when the problem only involves two agents, the approximation ratio can be improved to $2/3$.
It is noted that, to ensure any finite approximation of MMS fairness, significant amount of social welfare is inevitably sacrificed.

We then study EF1 allocation whose existence is guaranteed \cite{DBLP:conf/sigecom/LiptonMMS04}.
We prove that there exist instances for which none of the EF1 allocations can ensure better than $1/n$ fraction of the maximum social welfare.
But this result does not exclude the possibility of constant approximations for special cases.
In particular, we consider three cases: (1) binary-weight functions, (2) two-agents, (3) homogeneous-agents. 
For each setting, we design polynomial-time algorithms that compute EF1 allocations whose social welfare is at least a constant fraction of the maximum social welfare that can be achieved without fairness constraints.

\subsection{Related Works}

Two separate research lines are closely related to our work, namely graph partition and fair division. 


{\em Graph Partition.}
Partitioning graphs into  balanced subgraphs has been extensively studied in operations research \citep{DBLP:journals/eor/MiyazawaMOW21} and computer science \citep{DBLP:series/lncs/BulucMSS016}. There are several popular objectives for evaluating whether a partition is balanced. 
Among the most prominent ones are the max-min (or min-max) objectives,
where the goal is to maximize (or minimize) the total weight of the minimum (or maximum) part. 
Particularly, the vehicle routing problem (VRP) \citep{DBLP:journals/eor/KocBJL16a}, which generalizes the travelling salesperson problem (TSP), is closely related to our work.
It asks for an optimal set of routes for a number of vehicles, 
to visit a set of customers. 
There are a number of popular variants for the VRP, e.g., the so called heterogeneous vehicle routing problem
\citep{DBLP:journals/mp/Yaman06,DBLP:conf/swat/Rathinam0BS20}. 
There are many other combinatorial structures studied in graph partitioning problems. 
For example, in the min-max tree cover (a.k.a. nurse station location) problem, the task is to use trees to cover an edge-weighted graph such that the largest tree is minimized \citep{DBLP:journals/algorithmica/KhaniS14}. 
This problem also falls under the umbrella of a more general problem, the  graph covering problem, 
where a set of pairwise disjoint subgraphs (called templates) is used to cover a given graph, such as paths \citep{DBLP:journals/disopt/FarbsteinL15}, cycles \citep{DBLP:conf/ipco/TraubT20},  and matchings \citep{DBLP:journals/eor/KressMP15}.





{\em Fair Division.}
Allocating a set of indivisible items among multiple agents is a fundamental problem in the fields of multi-agent systems and computational social choice, and we refer the readers to recent surveys  \citep{DBLP:journals/corr/abs-2202-07551,DBLP:journals/corr/abs-2202-08713} for more detailed discussion.
Envy-freeness (EF) and maximin share fairness (MMS) are two well accepted and extensively studied solution concepts.
However, with indivisible items, these requirements are demanding and thus the state-of-the-art research mostly studies their relaxations and approximations. 
For example, EF1 allocation is studied as a relaxation of EF which always exists  \citep{DBLP:conf/sigecom/LiptonMMS04}. 
Various constant approximation algorithms for MMS allocations are proposed in \citep{kurokawa2018fair,garg2019improved} for additive valuations and in \citep{DBLP:journals/teco/BarmanK20,DBLP:conf/sigecom/GhodsiHSSY18} for subadditive valuations.
Our work focuses on indivisible graphical items where agents have combinatorial valuations (neither subadditive nor superadditive) depending on the structural properties.
Moreover, all the existing algorithms for non-additive valuations run in polynomial time only if the computation of valuations is assumed to be effortless (i.e., oracles).
In contrast, in this work, we aim at designing truly polynomial-time approximation algorithms without valuation oracles. 

\section{Preliminaries}\label{preliminaries}
Denote by $G=(V, E)$ an undirected graph with no reflexive edges, where $V$ contains all vertices and $E$ contains all edges. 
The vertices are the resources, also called items, that are to be allocated to $n$ heterogeneous agents, denoted by $N$.
Each agent $i$ has an edge weight function $w_i:E \to \bR^+\cup \{0\}$, which may be different from others'.
If $w_i(e) \in \{0,1\}$ for all $e\in E$, then the weight function is called binary. 
Let $\w=(w_1,\cdots,w_n)$.
A matching $M\subseteq E$ is a set of vertex-disjoint edges, and let $w_i(M) = \sum_{e\in M}w_i(e)$. Let $M(V)$ be the maximum (weighted) matching within the induced subgraph $G[V]$. 
For any subgraph $G'$, let $V(G')$ and $E(G')$ be the sets of vertices and edges in $G'$, respectively. 
An allocation $\bX=(X_1,\cdots,X_n)$ is a partition of $V$ such that $\cup_{i\in N}X_i = V$ and $X_i\cap X_j = \emptyset$ for $i\neq j$.
If $\bigcup_{i\in N}X_i \subsetneq V$, the allocation is called partial.
Each agent $i$ has a utility function $u_i: 2^V \to \bR^+\cup \{0\}$, where $u_i(X_i)$ equals the weight of a maximum (weighted) matching in $G[X_i]$. 
When the agents have identical valuations (i.e., homogeneous agents), 
we omit the subscript and use $w(\cdot)$ and $u(\cdot)$ to denote all agents' weight and utility functions. 
A problem instance is denoted by $\cI = (G, N)$. When we want to highlight the weight function, $w$ is also included as a parameter, i.e., $\cI = (G, N, w)$.


Next, we introduce the solution concepts. 
Our first fairness notion is {\em maximin share} (MMS) \citep{budish2011combinatorial}.
Letting $\Pi_n(V)$ be the set of all $n$-partitions of $V$, the maximin share of agent $i$ is
\[
\MMS_i(\cI) = \max_{\bX \in \Pi_n(V)} \min_{j \in N} u_i(X_j).
\]
We may write $\MMS_i$ for short if $\cI$ is clear from the context.
Therefore agent $i$ is satisfied regarding MMS fairness if her utility is no smaller than $\MMS_i$.


\begin{definition}[$\alpha$-MMS]
For any $\alpha\ge 0$, an allocation $\bX=(X_1,\cdots,X_n)$ is called {\em $\alpha$-approximate maximin share} ($\alpha$-MMS) fair if for all agents $i \in N$, 
\[
u_i(X_i) \ge \alpha \cdot \MMS_i.
\]
The allocation is called MMS fair if $\alpha = 1$.
 \end{definition}

The second fairness notion is about {\em envy-freeness} (EF).
An allocation $\bX$ is called EF if no agent envies any other agent's bundle, i.e.,
\[
u_i(X_i) \ge u_i(X_j) \text{ for all agents $i,j \in N$.}
\]
EF is very hard to satisfy; consider a simple example, where the graph is a triangle and two agents have weight 1 for all edges.
Then in every allocation, there is one agent who gets at most one vertex (with utility 0) and the other agent gets at least two vertices (which contains an edge and thus has utility 1). Accordingly, we focus on {\em envy-free up to one item} instead \citep{budish2011combinatorial}.
\begin{definition}[EF1]
An allocation $\bX=(X_1,\cdots,X_n)$ is called {\em envy-free up to 1 item} (EF1) if for any $i$ and $j$, there exists $g\in X_j$ such that $u_i(X_i) \ge u_i(X_j \setminus \{g\})$.
\end{definition}

Besides fairness, we also want the allocation to be efficient. 
Given an allocation 
$\bX=(X_1,\cdots,X_n)$, the {\em social welfare} of $\bX$ is $\sw(\bX) = \sum_{i\in N} u_i(X_i)$.
Note that given any instance $\cI$, the best possible social welfare of any allocation is the weight of a maximum matching in the graph $G$ by setting the weight of each edge to $\max_{i\in N} w_i(e)$, which is denoted by $\sw^*(\cI)$.
If the instance $\cI$ is clear from the context, we also denote $\sw^*(\cI)$ as $\sw^*$ for short. 

\section{MMS Fair Allocations}\label{Fair:identical}
\subsection{Heterogeneous Agents}
\begin{theorem}
\label{thm:heter:mms:negative}
No algorithm has bounded approximation guarantee for MMS, even if there are two agents with non-identical binary weight functions on the graph.
\end{theorem}
\begin{proof}
Consider the example as shown in Figure~\ref{fig1}. The graph containing four nodes $\{v_1, v_2, v_3, v_4\}$ is allocated to two agents whose valuations (i.e., edge weights) are shown in Figure \ref{figure202101053} and \ref{figure202101052} respectively. 
It can be verified that $\MMS_i=1$ for both $i=1,2$. 
However,  no matter how we allocate the vertices to the agents, one of them receives utility of $0$.
\begin{figure}[htb]
\begin{center}
\subfigure[Agent $1$'s Metric]{
\includegraphics[width=1.2in]{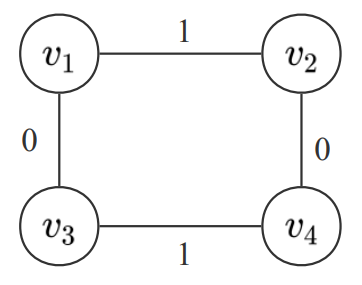}
\label{figure202101053}
}
\subfigure[Agent $2$'s Metric]{
\includegraphics[width=1.2in]{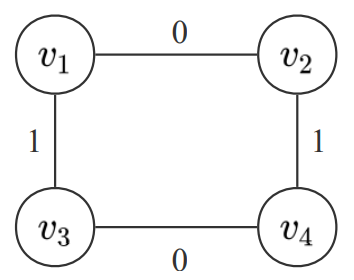}
\label{figure202101052}
}
\caption{A bad example for which no allocation has bounded approximation of MMS fairness.}
\label{fig1}
\end{center}
\end{figure}
\end{proof}
Theorem \ref{thm:heter:mms:negative} is very strong in the sense that it excludes the possibility of designing algorithms with bounded approximation ratio for MMS even for the special cases of two-agent or binary weight functions.

\subsection{Homogeneous Agents}

Due to the strong impossibility, we study the case of identical valuations, where MMS fairness degenerates to the max-min objective, where the problem is to partition a graph into $n$ subgraphs so that the smallest weight of the maximum matchings in these subgraphs is maximized. 
It is easy to see that finding such an allocation is NP-hard even when there are two agents and the graph only consists of independent edges, which is essentially a Partition problem. 
Thus, our target is polynomial-time approximation algorithms.
Without loss of generality, in this section, we assume $w(e) \ge 1$ for all $e\in E$.
Since the agents are identical, the subscript in $\MMS_i$ is omitted. 
Our main result in this section is as follows. 
\begin{theorem}
\label{thm:maxmin:n}
For homogeneous agents, we can compute a $1/8$-MMS allocation in polynomial time.
\end{theorem}


Given an instance $\cI=(G, N)$, to design such an algorithm with guaranteed approximation of MMS fairness, the thought at a glance is to allocate a maximum matching in $G$.
That is, we compute a maximum matching $M^* \subseteq E$, and then partition $M^*$ into $n$ bundles $(M_1,\cdots,M_n)$ where $w(M_1)\ge \cdots \ge w(M_n)$ such that $w(M_n)$ is as large as possible. 
However, such an allocation can be arbitrarily bad, let alone maximizing the minimum bundle being an NP-hard problem. 
Consider an example with two agents and the graph is shown in Figure~\ref{figure20210105} where $\Delta > 1$ is arbitrarily large. 
Any allocation with bounded approximation ratio of MMS fairness ensures that every agent has value 1, but by partitioning the maximum matching (which contains a single edge with weight $\Delta$) the smaller bundle has value 0.

\begin{figure}[h]
\centering
\includegraphics[width=3in]{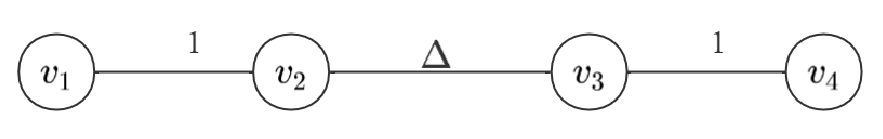} 
\caption{
A bad example when partitioning the maximum matching does not have bounded approximation of MMS.
}\label{figure20210105}
\end{figure}

Before describing our algorithm, we first define greedy partition of the maximum matching.
\paragraph{Greedy Partition.}
Given a matching $M$, partition $M$ into $\Gamma(M)=(M_1, \cdots, M_n)$ as follows.
\begin{itemize}
    \item Sort and rename the edges in $M$ such that $w(e_1)\ge\cdots\ge w(e_k)$ where $k=|M|$. 
    \item Initially set $M_1 = \cdots=M_n = \emptyset$. 
    \item For $i=1,\cdots, k$, select $j$ such that $w(M_j) \le w(M_{j'})$ for all $j'$ and set $M_j = M_j \cup \{e_i\}$.
    \item Sort and rename $M_1,\cdots,M_n$ so that $w(M_1) \ge \cdots \ge w(M_n)$.
    
\end{itemize}
The greedy partition $(M_1,\cdots,M_n)$ of the maximum matching $M^*$ corresponds to an allocation of vertices where unmatched vertices $V' = V\setminus \cup_{i\in N} V(M_i)$ can be allocated arbitrarily.
Although this allocation might not be good in general,
when the graph is unweighted ($w(e)=w(e')$ for all $e,e'\in E$) or $|M_1| \ge 2$, it ensures a good approximation.
\begin{lemma}
\label{lem:maxmin:unit}
If $G$ is unweighted, the greedy partition $(M_1,\cdots,M_n)$ of $M^*$
is an MMS allocation.
\end{lemma}
\begin{proof}
Without loss of generality, assume all edges have weight 1. In the greedy partition $(M_1,\cdots,M_n)$ of $M^*$, for any $i\in N$,
\[
|M_i| \ge |M_n| = \lfloor \frac{|M^*|}{n} \rfloor.
\]
Let $(O_1,\cdots,O_n)$ be an optimal max-min allocation.
If $\MMS = |M(O_n)| > |M_n|$, then for all $i\in N$,
\[
|M(O_i)| \ge \lfloor \frac{|M^*|}{n} \rfloor + 1.
\]
Thus
\begin{align*}
    \sum_{i\in N} |M(O_i)| \ge n \cdot \lfloor \frac{|M^*|}{n} \rfloor + n > |M^*|,
\end{align*}
which is a contradiction with $M^*$ being a maximum matching.
\end{proof}

\begin{lemma}\label{lemma20220105}
If $|M_1| \ge 2$, $\Gamma(M^*)$ corresponds to an allocation that is $1/2$-MMS fair.
\end{lemma}
\begin{proof}
Denote by $O=(O_1, O_2, \cdots, O_n)$ the optimal solution, where $u(O_1)\ge u(O_2) \ge \cdots \ge u(O_n)$ and $\MMS(\cI) = u(O_n)$. Under the maximum matching $M$, consider the greedy partition $(M_1, M_2, \cdots, M_n)$, where $u(M_1) \ge u(M_2) \ge \cdots \ge u(M_n)$. In greedy partition procedure, all edges are sorted in descending order of their weights and each time we select the edge with the largest weight in the remaining edge set and allocate it to the bundle with the least total utility.
If $|M_1| \geq 2$, consider the last edge $e$ added to $M_1$, we have $w(M_n) \ge w(e)$, since there exists at least one edge added to $M_n$ before edge $e$ is added to $M_1$. Since in the greedy procedure, edges are added to the bundle with least utility, we have $w(M_n) \ge w(M_1 \setminus \{e\})$. Furthermore, we have
\begin{equation}\nonumber
\begin{aligned}
w(M_n) &\geq \frac{1}{2}(w(e)+w(M_1/{e})) \geq \frac{1}{2}w(M_1)
\\& \ge \frac{1}{2n}\sum\limits_{i=1}^{n}w(M_i)  \ge \frac{1}{2n}\sum\limits_{i=1}^{n}u(O_i)\\& \ge
 \frac{1}{2}u(O_n),
     \end{aligned}
     \end{equation}
and the lemma holds accordingly.
\end{proof}
The tricky case is when $M_1$ contains a single edge $e^*$. 
The greedy partition fails because $w(e^*)$ is too large so that very few edges or edges with very small weights can be put in $M_i$'s for $i\ge 2$.
One way to overcome this difficulty is to decrease the weight of $e^*$ and re-compute a maximum matching, through which the advantage of $e^*$ can diminish.
For simplicity, assume all edge weights are powers of 2.
This is without much loss of generality which decreases the approximation ratio by at most $1/2$.
\begin{lemma}\label{lem:maxmin:rounding}
Let $\cI'=(G,N,w')$ be the instance obtained from $\cI=(G,N,w)$ by rounding all edge weights down to nearest powers of 2.
If $(X_1,\cdots,X_n)$ is an $\alpha$-MMS allocation of $\cI'$, it is also an $\alpha/2$-MMS allocation of $\cI$.
\end{lemma}
\begin{proof}
Let $\cI''=(G,N,w'')$ be the instance obtained from $\cI$ by halving all its edge weights. It is easy to see that 
\[
\MMS(\cI'')=\frac{1}{2}\cdot \MMS(\cI).
\]
Moreover, the weight of all edges in instance $\cI'$ is at least as large as that in instance $\cI''$, and thus 
\[
\MMS(\cI') \geq \MMS(\cI'')=\frac{1}{2}\cdot\MMS(\cI).
\]
Finally, since $u_i(X_i) \geq \alpha \cdot \MMS(\cI')$ for all $i \in N$, then 
\[
u_i(X_i) \geq \frac{\alpha}{2}\cdot \MMS(\cI), 
\]
and thus the lemma holds.
\end{proof}
Now we are ready to describe our Algorithm \ref{alg20210091}. 
We first compute a maximum matching $M^*$ and its greedy partition $\Gamma(M^*)=(M_1,\cdots, M_n)$ such that $w(M_1)\ge \cdots \ge w(M_n)$. 
If $|M_1|\ge 2$, by Lemmas \ref{lemma20220105} and \ref{lem:maxmin:rounding}, we can directly output the corresponding partition of vertices so that the approximation ratio is at least $1/4$.
If $|M_1| = 1$, we consider two cases. 
When $w(M_n) \ge 1/2\cdot w(M_1)$, $w(M_n)$ is still not too small and we can stop the algorithm with a constant approximation ratio. 
However, if $w(M_n) < 1/2\cdot w(M_1)$, it means the utility of the smallest bundle is much less than that of the largest bundle. 
Then we update the edge weights:
Let $H$ be the edges with weights no smaller than $w(e_1)$ where $e_1$ is the edge in $M_1$, and decrease all their weights to $1/2\cdot w(e_1)$.
By repeating this procedure, eventually we reach an allocation for which $w(M_n) \ge 1/2 \cdot w(M_1)$ or $|M_1| \ge 2$.
\begin{algorithm}
  \caption{Approximately MMS Fair Algorithm
    \label{alg20210091}}
  \begin{algorithmic}[1]
  \REQUIRE Instance $\cI=(G,N)$ with $G=(V, E; w)$.
  \ENSURE Allocation $\bX=(X_1, \cdots, X_n)$.
    \STATE 
    For all $e \in E$, reset
    \[
      w(e) = 2^{\lfloor \log w(e) \rfloor }.
    \]
      
   \STATE \label{step:maxmin:n}Find a maximum matching $M^*$ in $G$.
   Denote by $V'$ the set of unmatched vertices.
   \STATE Find the greedy partition $\Gamma(M^*)=(M_1,\cdots,M_n)$ of $M^*$ such that $w(M_1) \ge \cdots \ge w(M_n)$.
  \WHILE{$w(M_1) > 2\cdot w(M_n)$ and $G$ has different weights}
        \STATE Let $e_1$ be the edge in $M_1$ and $H = \{e\in E \mid w(e) \ge w(e_1)\}$.
        \STATE \label{step:maxmin:n:1} Let $w(e) = w(e_1)/2$ for all $e\in H$.
        \STATE Re-compute a maximum matching $M^*$.
        \STATE Re-set $V'$ to be unmatched vertices by $M^*$. 
        \STATE Re-compute the greedy partition $\Gamma(M^*)=(M_1,\cdots,M_n)$ such that $w(M_1) \ge \cdots \ge w(M_n)$.
  \ENDWHILE
    \STATE Set $X_i = V(M_i)$ for $i =1, \cdots, n-1$.
    \STATE Set $X_n = V(M_n) \cup V'$.
    \STATE Return allocation $(X_1, \cdots, X_n)$.
   
  \end{algorithmic}
\end{algorithm}
\begin{proof}[Proof of Theorem \ref{thm:maxmin:n}]
First, we show Algorithm \ref{alg20210091} is well-defined and runs in polynomial time. 
Every time when the condition of the {\bf while} loop holds, either the graph has different weights and an allocation is returned or the weights of the heaviest edges are decreased by $1/2^k$ with some $k\ge 1$. 
Thus the {\bf while} loop is executed $O(\max_{e\in E} \log w(e))$ rounds. 

Next we prove the approximation ratio. 
By Lemma \ref{lem:maxmin:rounding}, we only need to consider the instance where the edge weights are powers of 2 and show the allocation is 1/4-approximate MMS fair.
Denote by $O=(O_1,\cdots,O_n)$ the optimal solution, where $u(O_1)\ge \cdots\ge u(O_n)$ and $\MMS(\cI) = u(O_n)$.
The first time when we reach the {\bf while} loop, if $w(M_1) \le 2\cdot w(M_n)$, 
\[
w(M_n) \ge \frac{1}{2}\cdot w(M_1) \ge \frac{1}{2}\cdot u(O_n)= \frac{1}{2}\cdot \MMS(\cI),
\]
where the second inequality holds because $M^*$ is a maximum matching in $G$.
Thus the allocation is 1/2-MMS. 
If all edges have the same weight, then by Lemma \ref{lem:maxmin:unit}, the allocation is optimal. 

We move into the {\bf while} loop if $w(M_1) > 2\cdot w(M_n)$ and the edge weights are not identical.
Note that $w(M_1) > 2\cdot w(M_n)$ implies $M_1$ contains a single edge denoted by $e_1$.
Otherwise consider the last edge added to $M_1$ in the greedy partition, denoted by $e'$.
Then $w(M_1\setminus\{e'\}) \le w(M_n)$ and $w(e') \le w(M_n)$, which implies $w(M_1) \le 2\cdot w(M_n)$.
After the {\bf while} loop, denote by $\cI'$ the instance, by $w'(\cdot)$ the new weights with new utility function $u'(\cdot)$, by $O'=(O'_1,\cdots,O'_n)$ the new optimal solution and by $M'$ the maximum matching with greedy partition $(M'_1,\cdots,M'_n)$. 
Then we have the following claim. 
\begin{claim}\label{claim:maxmin:while}
After each {\bf while} loop, one of the following two cases holds.
\begin{itemize}
    \item Case 1. $w(e_1) \ge 2\cdot\MMS(\cI)$, then $\MMS(\cI')=\MMS(\cI)$;
    \item Case 2. $w(e_1) < 2\cdot\MMS(\cI)$, then $2\cdot \MMS(\cI')>\MMS(\cI)$ and $w'(M'_1) \le 2\cdot w'(M'_n)$.
\end{itemize}
\end{claim}
\begin{proof}
We first consider Case 1. 
For any $O_i$, if $w(e) < w(e_1)$ for all $e\in M(O_i)$, then $u(O_i)$ does not decrease. 
If $w(e) \ge w(e_1)$ for some $e\in M(O_i)$, $u(O_i) \ge w(e_1) \ge 2\cdot\MMS(\cI)$ and after decreasing the weights to $ w(e_1)/2$, $u'(O_i) \ge \MMS(\cI)$,
implying the existence of an allocation with the minimum utility no smaller than $\MMS(\cI)$, which means $\MMS(\cI') = \MMS(\cI)$.  

Second, we consider Case 2 when $w(e_1) < 2\cdot\MMS(\cI)$. It is straightforward that $2\cdot \MMS(\cI') > \MMS(\cI)$ since $w(e_1) > \MMS(\cI)$ (If $w(e_1)=\MMS(\cI)$, then $w(e_1)=w(M_1)= \cdots =w(M_n)=\MMS(\cI)>\frac{1}{2}w(M_n)$, since $M^*$ is a maximum (weighted) matching. The Algorithm \ref{alg20210091} moves out of the {\bf while} loop and returns the optimal solution), and after decreasing the weights of some edges to $w(e_1)/2$, $u'(O_i) \ge w(e_1)/2 > \MMS(\cI)/2$ (If $e_1 \in M(O_i)$, $u'(O_i) \ge w(e_1)/2$. Otherwise, $u'(O_i)=u(O_i) \ge \MMS(\cI) > \MMS(\cI)/2$ ). 
Next we show $|M'_1| \ge 2$ which implies $w'(M'_1) \le 2\cdot w'(M'_n)$. 
For the sake of contradiction, assume $M'_1 = \{e'_1\}$. 
Therefore, we derive $w'(e'_1) = w'(M'_1) \ge \cdots \ge w'(M'_n)$,
and thus
\[
w'(M') \le n \cdot w'(e'_1) < n\cdot w'(O'_n) \le \sum_{i\in N} w'(O'_i).
\]
This is a contradiction with $M'$ being a maximum matching in $G$, which completes the proof.
\end{proof}

By Claim \ref{claim:maxmin:while}, the {\bf while} loop will not execute Case 2 or it executes Case 1 for several times and then Case 2 for exactly once. 
If Case 2 is not executed, then the allocation is 1/2-MMS fair and the analysis is the same with the case when the {\bf while} loop is not executed. 

If Case 2 is executed once, then by Claim \ref{claim:maxmin:while},
\[
w'(M'_n) \ge \frac{1}{2}\cdot w'(M'_1) \ge  \frac{1}{2}\cdot  \MMS(\cI') \ge \frac{1}{4}\cdot \MMS(\cI).
\]

Finally, by Lemma \ref{lem:maxmin:rounding}, the allocation is 1/8-MMS for any instance with arbitrary weights.
\end{proof}

When $n=2$, we can improve Algorithm \ref{alg20210091} and obtain a better approximation ratio of 2/3. Algorithm \ref{alg20210811} is similar with Algorithm \ref{alg20210091}; we first compute a maximum matching $M^*$ and a max-min partition $(M_1,M_2)$ with $w(M_1)\ge w(M_2)$. 
If $w(M_1)>2w(M_2)$, we output the corresponding allocation.
Otherwise, in graph $G$, we directly delete the edge that $M_1$ contains. 
We repeat the above procedure until all edges are removed. 
\begin{algorithm}[h]
  \caption{Max-Min Allocation for 2 Agents
    \label{alg20210811}}
  \begin{algorithmic}[1]
    \REQUIRE Instance $\cI=(G,N,u)$ with $G=(V, E; w)$.
  \ENSURE Allocation $\bX=(X_1,  X_2)$.
  \STATE \label{step:maxmin:n20220108}Find a maximum matching $M^*$ in $G$.
   Denote by $V'$ the set of unmatched vertices by $M^*$.
   \STATE Find the greedy partition $(M_1,M_2)$ of edges in $M^*$ such that $w(M_1) \ge w(M_2)$.
    \STATE Let $Max=w(M_2)$.
    \STATE Set $X_1= V(M_1)$ .
    \STATE Set $X_2 = V(M_2) \cup V'$.
  \WHILE{$w(M_1) > 2w(M_2)$}\label{step202201091}
    \STATE {\color{gray} // Note that if the while loop is entered, then $|M_1|=1$.}
  \STATE By Lemma \ref{lemma20220105}. $M_1$ must contain only one edge. Suppose $M_1=\{e^*\}$.
         \STATE Delete edge $e^*$.
         \STATE Re-compute a maximum matching $M^*$.
        \STATE Re-set $V'$ to be unmatched vertices by $M^*$. 
        \STATE Re-compute the greedy partition $(M_1,M_2)$ of $M^*$ such that $w(M_1) \ge  w(M_2)$.
        \IF{$Max<w(M_2)$}
          \STATE $Max=w(M_2)$.
          \STATE Set $X_1= V(M_1)$.
          \STATE Set $X_2 = V(M_2) \cup V'$.
        \ENDIF
  \ENDWHILE\label{step202201092}
 \STATE Output allocation $(X_1, X_2)$.
  \end{algorithmic}
\end{algorithm}
\begin{theorem}
\label{thm:maxmin:2}
Algorithm \ref{alg20210811} outputs an allocation that is $2/3$-approximate max-min fair in polynomial time .
\end{theorem}
\begin{proof}
Given an Instance $\cI=(G, N, u)$ with $G=(V, E; w)$. Denote by $O=(O_1, O_2)$ the optimal solution, where $u(O_1)\ge u(O_2)$ and $\MMS(\cI) = u(O_2)$.
The first time when we reach the {\bf while} loop, if $w(M_1) \le 2\cdot w(M_2)$, allocation $(X_1, X_2)=(V(M_1), V(M_2) \cup V')$ has been output.
By Algorithm \ref{alg20210811}, we have 
\begin{equation}\nonumber
    w(M_1) \geq u(O_1) \geq u(O_2) \geq w(M_2).
\end{equation}
Moreover, 
\begin{equation}\label{equation20211204}\nonumber
\begin{aligned}
w(M_2) &\geq \frac{1}{3}\cdot(w(M_1)+w(M_2)) \\&\geq \frac{1}{3}\cdot(u(O_1)+u(O_2)) \\ & \geq  \frac{1}{3} \cdot 2\cdot u(O_2)=\frac{2}{3}\cdot u(O_2).
\end{aligned}
\end{equation}
We move into the {\bf while} loop if $w(M_1) > 2\cdot w(M_2)$. In such case, $M_1$ contains only one edge, i.e., $|M_1|=1$. Suppose $M_1=\{e^*\}$. There are two subcases:
\begin{itemize}
    \item Case 1\label{lemmacase1}: $e^* \in M(O_1) \cup M(O_2)$ 
    \item Case 2: $e^* \notin M(O_1) \cup M(O_2)$
\end{itemize}
First. consider $|M_1|=1$ and $e^* \in M(O_1) \cup M(O_2)$. We have: $e^* \in M(O_1)$ and $w(M_2)=u(O_2)=\MMS(\cI)$. The optimal solution has been found and recorded. Therefore, the approximation ratio of the max-min partition is $1$. Then the {\bf while} loop is executed for the next round. When $e^* \notin M(O_1) \cup M(O_2)$, edge $e^*$ is deleted. Let $(M'_1, M'_2)$ be the greedy partition after deleting edge $e^*$. Then there are two subcases:
\begin{itemize}
    \item Subcase 1: $|M'_1| \geq 2$
    \item Subcase 2: $|M'_1| = 1$
\end{itemize}
For Subcase 1, we will get out of the {\bf while} loop and a $2/3$-approximate max-min allocation has been determined. For Subcase 2, the {\bf while} loop is executed for the next round. 
Since we are not sure whether $e^* \in M (O_1)\cup M(O_2)$, the {\bf while} loop is executed for at most $O(|V|^2)$ rounds. The output $(X_1, X_2)$ is at least $2/3$-approximate max-min fair allocation.
Thus, the theorem holds.
\end{proof}
\begin{lemma}
Algorithm \ref{alg20210811} outputs an allocation that is $1/2$-approximate max-min fair  by eliminating at most two edges.
\end{lemma}
\begin{proof}
Given an Instance $\cI=(G, N, u)$ with $G=(V, E; w)$. Denote by $O=(O_1, O_2)$ the optimal solution before eliminating any edge, where $u(O_1)\ge u(O_2)$ and $\MMS(\cI) = u(O_2)$.
Initially, under the maximum matching $M'$, we find the greedy max-min partition $(M_1, M_2)$ such that $w(M_1) \geq w(M_2)$. If $w(M_1) \le 2\cdot w(M_2)$, similar to the proof of Theorem \ref{alg20210811}, $(M_1, M_2)$ is a $2/3$-approximation max-min partition. Next, consider that $M_1$ contains only one edge. Suppose $M_1=\{e_1\}$ and $e_1 \notin M(O_1) \cup M(O_2)$ (otherwise, by Case 1 in Theorem \ref{thm:maxmin:2}, $w(M_2)=u(O_2)=\MMS(\cI)$. The optimal solution has been found). If we eliminate edge $e_1$, under the re-computed maximum matching, we find the greedy max-min partition $(M_1^{'} , M_2^{'})$.  Let $O'=(O'_1, O'_2)$ denote the optimal solution after eliminating edge $e_1$, where $u(O'_1)\ge u(O'_2)$ and $\MMS(\cI') = u(O'_2)$. $M^{'}_1=\{e^{'}_1\}$ and $e^{'}_1 \notin M(O'_1) \cup M(O'_2)$. We first show that the two edges $e_1$ and $e^{'}_1$ have one common endpoint. By Algorithm \ref{alg20210811}, 
\begin{equation}\nonumber
  w(e_1)>u(O_1)\geq u(O_2)>w(M_2),
\end{equation}
 and 
 \begin{equation}\nonumber
 w(e_1^{'})>u(O'_1)\geq u(O'_2)>w(M_2^{'}).
\end{equation}
 Since $e_1 \notin M(O_1) \bigcup M(O_2)$, eliminating edge $e_1$ does not change the optimal solution, i.e., $O=O'$. Hence, $w(e_1^{'})>w(M_2)$. Furthermore, 
 \begin{equation}\nonumber
     w(e_1)+w(e_1^{'})>w(e_1)+w(M_2).
 \end{equation}
 Therefore, edges $e_1$ and $e_1^{'}$ make a maximum matching, which implies that there exists an allocation to improve the max-min value from $u(O_2)$  to $w(e'_1)$. It results in a contradiction. Hence, the two edges $e_1$ and $e^{'}_1$ have one common endpoint. Therefore, there are two cases. First, we consider Case 1 (as shown in Figure \ref{figure202108221}): suppose $a_1$ and $a_2$ are two edges in $M(O_2)$, and without loss of generality, $w(a_1) \geq w(a_2)$. The two endpoints of edge $a_i, i \in [2]$ is denoted by $e_1(a_i)$ and $e_2(a_i)$. Denote $O_2^{*}=O_2/\{e_1(a_1), e_2(a_1), e_1(a_2), e_2(a_2)\}$. 
 Thus
 \begin{equation}\nonumber
 \max({w(a_1), w(a_2)}) \geq \frac{1}{2}(w(a_1)+w(a_2)),
 \end{equation}
and then
  \begin{equation}\nonumber
 w(a_1)+u(O_2^{*}) \geq \frac{1}{2}u(O_2).
 \end{equation}
What's more
\begin{equation}\nonumber
w(M_2^{'}) \geq w(a_1)+u(O_2^{*})
\end{equation}
holds; otherwise, replacing $M_2^{'}$ with $\{a_1\} \cup M(O_2^{*})$ can make a matching with larger welfare. Therefore, $w(M_2^{'}) \geq 1/2\cdot u(O_2)$. Next, we consider Case 2 (as shown in Figure \ref{figure202108222}). Suppose after the two edges $e_1$ and $e_1^{'}$ have been deleted, under the new maximum matching $M''$, we find the greedy partition $(M_1^{''}, M_2^{''})$. If $w(M''_1) \le 2 \cdot w(M''_2)$, by Theorem \ref{thm:maxmin:2}, the $2/3$-approximate max-min partition can be found. Otherwise, $|M_1^{''}| = 1$, assume $M_1^{''}=\{e''_1\}$. Without loss of generality, suppose edge $e''_1$ shares one common point with edge $a_2$. Let $O''_2=O_2 \setminus \{e_1(a_1), e_2(a_1), e_1(a_2), e_2(a_2)\}$, then 
\begin{equation}\nonumber
w(M_2^{''}) \geq w(a_1)+u(O''_2)\geq \frac{1}{2}u(O_2).
\end{equation}
Thus a $1/2$-approximate max-min partition has been found, and the lemma holds.
\end{proof}
\begin{figure}[htb]
\centering
\includegraphics[width=2.0in]{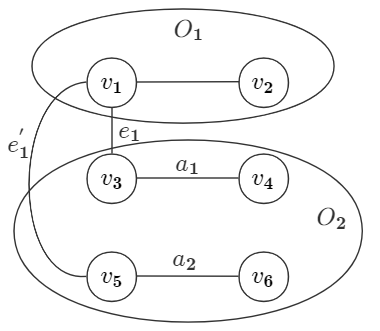}
\vspace{0.1in}
\caption{The graph is allocated to two agents with identical valuations. Two large edges $e_1$ and $e'_1$ have one common endpoint and their other endpoints are incident to two distinct edges.}\label{figure202108221}
\end{figure}
\begin{figure}
\centering
\includegraphics[width=2.0in]{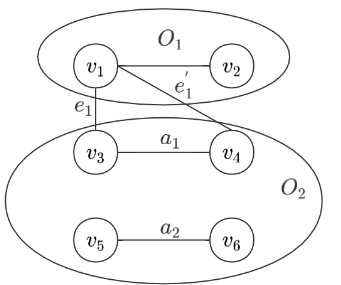}
\vspace{0.1in}
\caption{The graph is allocated to two agents with identical valuations. Two large edges $e_1$ and $e'_1$ have one common endpoint and their other endpoints are connected by an edge.}\label{figure202108222}
\end{figure}

\section{EF1 Allocations}
Recall the example in Figure \ref{figure20210105}.
The maximum social welfare is $\sw^* = \Delta$, but any bounded-approximate MMS allocation has social welfare $2 \ll \Delta$, which means to ensure (approximate) MMS, we have to sacrifice a significant amount of efficiency. 
Thus in this section, we turn to study EF1 allocations, whose existence is guaranteed \citep{DBLP:conf/sigecom/LiptonMMS04}.
An arbitrary EF1 allocation does not have any social welfare guarantee, and our goal in this section is to compute an EF1 allocation that also preserves high social welfare.
\subsection{General Heterogeneous Agents}\label{fair:heterogeneous}

Unfortunately, for the general case, we found that EF1 allocations cannot have good social welfare either. Note that the optimal social welfare $\sw^*$ is no longer the maximum matching under a single metric, which can be computed by
\[
\sw^* = \max_{X \in \Pi_n(V)} \sum_{i\in N}u_i(X_i).
\]
We have the following theorem.
\begin{theorem}
\label{thm:heter:ef1:negative}
For heterogeneous agents, no EF1 allocation can guarantee better than $1/n$ fraction of the optimal social welfare without fairness constraints. 
\end{theorem}
\begin{proof}
Now, we give an instance where, for any $\epsilon>0$, every EF1 allocation has social welfare at most $(1/n+\epsilon)\cdot \sw^*$.
If $\epsilon \ge 1 -1/n$, it holds trivially since no allocation can have social welfare more than $\sw^*$.
In the following, we assume $\epsilon < 1 -1/n$.

\begin{figure}[h]
\centering
\includegraphics[width=1.5in]{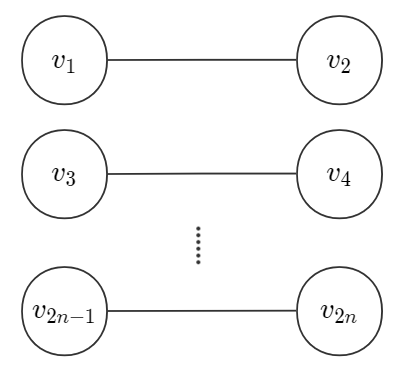}
\caption{A graph with $n$ disjoint edges is allocated to $n$ agents.}\label{figure20210106}
\end{figure}

Consider a graph with $n$ disjoint edges, as shown in Figure \ref{figure20210106}, which is to be allocated to $n$ agents. For each edge, agent $1$ has value~1, and the other agents have value~$\epsilon$. The maximum social welfare $\sw^*=n$ is achieved by allocating all edges to agent 1.
However, to guarantee EF1, at most one edge can be given to agent~1. The maximum welfare of an EF1 allocation is therefore at most $1+(n-1)\cdot \epsilon$ (each agent receives exactly one edge). The largest ratio is
\[
\frac{1+(n-1)\cdot \epsilon}{n}<\frac{1}{n}+\epsilon,
\]
which completes the proof of the Theorem, since $\epsilon$ can be arbitrarily small constant.\end{proof}
We now present a polynomial-time EF1 algorithm that achieves $\Omega(1/n^2)$ approximation of the social welfare for the general case.
\begin{theorem}\label{ef1:theorem:arbitrary}
For any instance $\cI=(G,N)$, Algorithm~\ref{alg:EF1:arbitr} returns an EF1 allocation with social welfare at least $1/(4n^2) \cdot \sw^*(\cI)$ in polynomial time.
\end{theorem}
\begin{algorithm}[htb]
  \caption{Computing EF1 Allocations for $n$ Heterogeneous Agents with Distinct Weights
    \label{alg:EF1:arbitr}}
  \begin{algorithmic}[1]
  \REQUIRE Instance $\cI=(G,N,w)$ with $G=(V, E)$. 
  \ENSURE Allocation $\bX=(X_1, \cdots, X_n)$.
  \STATE Initialize $X_i \leftarrow \emptyset, i \in N$. 
  Let $M_i$ be the maximum matching in $G[X_i]$ for agent $i$. 
  Denote by $\cG'=(N, \cE)$ the envy-graph on $\bX$. 
  \STATE Let $P = V\setminus (X_1\cup\cdots\cup X_n)$ be the set of unallocated items (called {\em pool}). 
  \STATE Denote $H$ as the set of agents who are not envied by any other agents.  Initialize $H \leftarrow N$.
  \STATE Let agent $i^*$ determine a maximum matching $M_{i^*}$ in graph $G$. Denote by $R$ the set of the remaining edges within the maximum matching.
  Initialize $R \leftarrow M_{i^*}$.
  \STATE Sort the edges $e \in M_{i^*}$ by non-increasing order according to their weight to agent $i^*$.
  \STATE Let agent $i^*$ pick one edge with largest weight $w_{i^*}(e)$ (with ties broken arbitrarily). 
  \WHILE{\{$R \neq \emptyset$\}} \label{ef1:arbitrar:while:first}
  \STATE Select one agent $i \in H$.
  \IF{\{$P_i=\emptyset$\} }
          \STATE Select one edge $e \in R$ with largest weight to agent $i^*$. Give one endpoint $v_1$ of $e$ to agent $i$ and put another endpoint $v_2$ in the corresponding pool $P_i$, i.e., $R \leftarrow R \setminus \{e\}$, $X_i \leftarrow X_i \cup \{v_1\}$, $P_i \leftarrow P_i \cup \{v_2\}$, $P \leftarrow P \setminus \{v_1, v_2\}$.
          \STATE Update the envy-graph $\cG'$ and set  $H$.
    \ELSE
         \STATE Give the node $v \in P_i$ to agent $i$, i.e., $P_i \leftarrow \emptyset$, $X_i \leftarrow X_i \cup \{v\}$.
         \STATE Update the envy-graph $\cG'$ and set $H$.
         \ENDIF 
  \ENDWHILE
  \STATE Return all the vertices within $P_i$ to the pool $P$, i.e., $P \leftarrow P \bigcup_{i \in N}P_i$.
 \STATE Execute the envy-cycle elimination procedure running on the remaining items $P$.
\STATE Return the allocation $(X_1, \cdots, X_n)$.
  \end{algorithmic}
\end{algorithm}
Without loss of generality, we assume $i^*$ to be the agent who has the maximum value of $u_i(V), i \in N$. Denote by $S=(S_1, \cdots, S_n)$ the partial allocation when we first move out of the {\bf while} loop in Step \ref{ef1:arbitrar:while:first}. Before presenting the proof of Theorem \ref{ef1:theorem:arbitrary}, we first present a useful lemma.
\begin{lemma}\label{ef1:lemma:0.5}
$\sum\limits_{i \in N}u_{i^*}(X_i) \geq \frac{1}{2}u_{i^*}(V)$.
\end{lemma}
\begin{proof}
     During the execution of Algorithm \ref{alg:EF1:arbitr}, there is at most one node $v_i$ in pool $P_i, i\in [n]$. Consider $P_i$ when we first move out of the {\bf while} loop in Step \ref{ef1:arbitrar:while:first}. If $P_i \neq \emptyset$, w.l.o.g. suppose $v_i \in P_i$ is one endpoint of edge $e_i \in M_{i^*}$. Let $N_1$ be the set of agents such that $P_i \neq \emptyset$.  Since the edges are picked by non-increasing order of their weight to agent $i^*$, we have $u_{i^*}(S_i) \geq w_{i^*}(e_i), i \in N_1$. Furthermore, we have
 \begin{equation*}
     2u_{i^*}(S_i) \geq u_{i^*}(S_i)+w_{i^*}(e_i) \geq u_{i^*}(S_i \cup \{v_i\}).
 \end{equation*}
 Thus, $u_{i^*}(S_i) \geq \frac{1}{2}u_{i^*}(S_i \cup \{v_i\}), i \in N_1$.
We have
\begin{align*}
    \sum\limits_{i \in N}u_{i^*}(X_i) & \geq  \sum\limits_{i \in N}u_{i^*}(S_i)\\ &  \geq \frac{1}{2}\sum\limits_{i \in N_1}u_{i^*}(S_i \cup \{v_i\})+\sum\limits_{i \in N \setminus N_1}u_{i^*}(S_i) \\ &  \geq \frac{1}{2}\sum\limits_{i \in N_1}u_{i^*}(S_i \cup \{v_i\})+\frac{1}{2}\sum\limits_{i \in N \setminus N_1}u_{i^*}(S_i)\\ &= \frac{1}{2}\sum\limits_{i\in N}u_{i^*}(V),
\end{align*}
where the last equality holds because the nodes within the remaining pool $P$ after we move out of the {\bf while} loop in Step \ref{ef1:arbitrar:while:first} do not have any effect on the maximum matching $M_{i^*}$.
We complete the proof of Lemma \ref{ef1:lemma:0.5}.
\end{proof}
Now we are ready to prove Theorem \ref{ef1:theorem:arbitrary}.
\begin{proof}[Proof of Theorem \ref{ef1:theorem:arbitrary}]
Let $N_e$ be the set of agents that agent $i^*$ envies. Since Algorithm \ref{alg:EF1:arbitr} admits EF1, there exists one node $v \in X_i, i \in N_e$ such that $u_{i^*}(X_{i^*}) \geq u_{i^*}(X_i \setminus \{v\}), i \in N_e$. For any agent $i \in N_e$, w.l.o.g. assume $v_1$ is the node such that $u_{i^*}(X_{i^*}) \geq u_{i^*}(X_i \setminus \{v_1\})$ and  $v_1, v_2 \in S_i$ are two endpoints of the edge $e \in M_{i^*}$. Since agent $i^*$ first picks the edge with the largest weight to itself, we have $u_{i^*}(X_{i^*}) \geq w_{i^*}(e)=u_{i^*}(\{v_1, v_2\})$. By the definition of EF1, $u_{i^*}(X_{i^*}) \geq u_{i^*}(X_i \setminus \{v_1\}) \geq u_{i^*}(X_i \setminus \{v_1, v_2\})$ holds. Thus, we have $u_{i^*}(X_{i^*}) \geq \frac{1}{2}(u_{i^*}(X_i \setminus \{v_1, v_2\})+u_{i^*}(\{v_1, v_2\}))=\frac{1}{2}u_{i^*}(X_i)$. Let $(X^*_1, \cdots, X^*_n)$ be the welfare maximization allocation, where $X^*_i, i \in N$ is the set of vertices allocated to agent $i$.
Therefore
\begin{align*}
    n \cdot u_{i^*}(X_{i^*}) & \geq \frac{1}{2} \sum\limits_{i \in N} u_{i^*}(X_i)\geq \frac{1}{2} \cdot \frac{1}{2} u_{i^*}(V)\\ &\geq \frac{1}{4} \cdot \frac{1}{n}\sum\limits_{i \in N}u_i(V)\geq \frac{1}{4n}\sum\limits_{i \in N}u_i(X_{i}^*)\\&=\frac{1}{4n}\sw^*(\cI),
\end{align*}
where the second inequality follows by Lemma \ref{ef1:lemma:0.5} and the third inequality holds because of the assumption that $i^*$ is the agent with the largest value of $u_i(V), i \in N$. Since in each iteration one node is allocated to an agent, the time complexity of Algorithm \ref{alg:EF1:arbitr} is at most $O(|V|^2)$, completing the proof of Theorem \ref{ef1:theorem:arbitrary}.
\end{proof}
Due to this hardness result of general case, in the following three sub sections, we present three cases when EF1 allocations manage to ensure constant fraction of the optimal social welfare.
\subsection{Binary Weight Functions}
\begin{algorithm}[!htb]
  \caption{Computing EF1 Allocations for $n$ Heterogeneous Agents with Binary Weights
    \label{alg:EF1:n04091}}
  \begin{algorithmic}[1]
  \REQUIRE Instance $\cI=(G,N,\w)$ with $G=(V, E)$.
  \ENSURE Allocation $\bX=(X_1, \cdots, X_n)$.
  \STATE Initialize $X_i \leftarrow \emptyset, i \in N$.  Let $M_i$ be the maximum matching in $G[X_i]$ for agent $i$. 
  Denote by $\cG'=(N, \cE)$ the envy-graph on $\bX$. 
  \STATE Let $P = V\setminus (X_1\cup\cdots\cup X_n)$ be the set of unallocated items (called {\em pool}).  
  
  
  \STATE Partition agents $i \in N$ into $k$ groups ${\bf A}(\bX)=(A_1, \cdots, A_k)$ such that agents in the same group have the same value, i.e., $u_i(X_{i}) = u_j(X_j)$ for $i,j \in A_l$ and $l\in [k]$.
  Assume $A_l$'s are ordered, i.e., $u_{i}(X_{i})< u_{j}(X_{j})$ for agents $i \in A_{t_1}$, $j \in A_{t_2}$ and $t_1 < t_2$.
  \STATE Let $t \leftarrow 1$ and $\tau \leftarrow |{\bf A}|$.
  \WHILE{\{$t\leq \tau$\}}\label{step:20220411}
  \STATE {\color{gray} // Case 1. Directly Allocate}
  \IF{there exists an agent $i \in A_{t}$ such that (1) there is an edge $e$ in $G[P]$ with $w_i(e)=1$ and (2) allocating the two endpoints $v_1, v_2$ of $e$ to agent $i$ does not break EF1 }
        \STATE 
        $X_i \leftarrow X_i \cup \{v_1,v_2\}$, $P \leftarrow P \setminus \{v_1,v_2\}$. \label{Allo:single:egde:ef1}
         \STATE Update $u_i(X_i)$ for $i \in N$ and the envy-graph $\cG'$.
        \STATE Update the partition of agents in $\bf{A}$.
        \STATE Reset $t \leftarrow 1$ and $\tau \leftarrow |\bf{A}|$.
    \STATE {\color{gray} // Case 2. Exchange and Allocate}
    \ELSIF{there exists agent $j \in N$ and $i \in A_t$ such that $j$ envies $i$ and there exists a subset with minimum size $V^* \subseteq P $ in graph $G$ such that $u_i(V^*) = u_i(X_i)$} \label{exchanging:from:ef1}
     \STATE Let $V^* \subseteq P $ be a set with minimum size such that $u_i(V^*) = u_i(X_i)$.
     \STATE Let $V^*_j \subseteq X_i$ be a set with minimum size such that $u_j(V^*_j)=u_j(X_j)+1$.
     
     \STATE  
     $P \leftarrow (P \setminus V^*) \cup X_j \cup (X_i \setminus V^*_j)$.
     \STATE
     $X_i \leftarrow V^*$, 
     $X_j \leftarrow V^*_j$.\label{exchanging:to:ef1}
    

        \STATE Update $u_i(X_i)$ for $i \in N$ and the envy-graph $\cG'$.
        \STATE Update the partition of agents in $\bf{A}$.
        \STATE Reset $t \leftarrow 1$ and $\tau \leftarrow |\bf{A}|$.
         \ELSE
         \STATE {\color{gray} // Case 3. Skip the Current Agent}
         \STATE $t \leftarrow t+1$.
    \ENDIF 
    \ENDWHILE \label{Alg:end_while_1}
    
    
    \STATE Execute the envy-cycle elimination procedure on the remaining items $P$. \label{Alg:binary:envy-cycle}
\STATE Return the allocation $(X_1, \cdots, X_n)$. 
  \end{algorithmic}
\end{algorithm}

We first show that if the agents have binary weight functions, we can compute an EF1 allocation whose social welfare is at least 1/3 fraction of the optimal social welfare.
Before introducing our algorithm, we recall the {\em envy-cycle elimination algorithm} proposed by \cite{DBLP:conf/sigecom/LiptonMMS04}, which always returns an EF1 allocation.
Given a (partial) allocation $(X_1,\cdots,X_n)$, we construct the corresponding {\em envy graph} $\cG'=(N, \cE)$, where the nodes are agents (and thus are used interchangeably) and  there is a directed edge from agent $i$ to agent $j$ if and only if $u_i(X_i) < u_i(X_j)$.
The {\em envy-cycle elimination algorithm} runs as follows.
We first find an agent who is not envied by the others, and allocate a new item to her.
If there is no such an agent, there must be a cycle in the corresponding envy graph. 
Then we resolve this cycle by reallocating the bundles: every agent gets the bundle of the agent that she envies in the cycle. 
We repeat resolving cycles until there is an unenvied agent.
The above procedures continue until all the items are allocated.
Note that in the execution of the algorithm, the agents' utilities can only increase, and  the returned allocation is EF1.

It is not hard to verify that the envy-cycle elimination algorithm does not have any social welfare guarantee, which is illustrated by the following example.  Consider a path of four nodes $v_1 \to v_2 \to v_3 \to v_4$, and two agents have the same weight 1 on all three edges $(v_1,v_2)$, $v_2,v_3$ and $v_3,v_4$.
    By ency-cycle elimination algorithm, we may first allocate the items in the following order: $v_1$ to agent 1, then $v_2$ to agent 2, then $v_3$ to agent 1 and finally $v_4$ to agent 2.
    Note that $u_1(\{v_1,v_3\}) = v_2(\{v_2, v_4\}) = 0$, however, the optimal social welfare is 2 by allocating $\{v_1,v_2\}$ to agent 1 and $\{v_3,v_4\}$ to agent 2. 
    Thus the approximation ratio of the social welfare is unbounded. There are several reasons. First, the algorithm does not control which item should be allocated to the unenvied agent so that the agent may receive a set of independent vertices. Second, once an item is allocated it cannot be recalled so that we are not able to revise any bad decision we have made.
To increase the social welfare, in each round of our algorithm, we try to allocate an edge (i.e., two items) to the agent $i$ with the smallest value so that the social welfare can increase by 1. 
However, we need to be very careful by allocating two items which may break the EF1 requirement even if $i$ is not envied by the others.
If allocating an edge $e$ to $i$ makes some agents $j$ envy $i$ for more than one item, we check whether $i$ can maintain her utility by selecting a bundle from unallocated items.
If so, we execute {\em exchange} procedure by asking $j$ to (properly) select a bundle from $X_i$ and $i$ to (properly) select a bundle from unallocated items so that the social welfare is increased by 1.
All the items in $X_i$ and the items in $X_j$ that are not selected by $i$ are returned to the algorithm.
If not, we try to allocate an edge to the agent with the second smallest value by executing the above procedures, and so on.
The description is in Algorithm \ref{alg:EF1:n04091} and we have the following.
\begin{theorem}\label{ef1:binary:theorem}
For any instance $\cI=(G,N)$ with binary weights, Algorithm \ref{alg:EF1:n04091} returns an EF1 allocation in polynomial time with social welfare at least $1/3 \cdot \sw^*(\cI)$.
\end{theorem}
Before proving Theorem \ref{ef1:binary:theorem}, we give some useful lemmas.
\begin{lemma}\label{lemma;ef1;binary}
During the execution of Algorithm \ref{alg:EF1:n04091}, the partial allocation maintains EF1.
\end{lemma}
\begin{proof}
In the execution of Algorithm \ref{alg:EF1:n04091}, two cases within the {\bf while} loop can change the partial allocation:
\begin{itemize}
    \item Case 1. Directly Allocate;
    \item Case 2. Exchange and Allocate.
\end{itemize}
Consider an arbitrary round $t \ge 1$.
In Case 1, a single edge is allocated to one agent if and only if such allocation still guarantees EF1. 
Now, we consider Case 2. If $i \in A_t$ is the agent who 
is able to pick a subset $V^* \subseteq P$ to maintain his own utility, i.e., $u_i(V^*)=u_i(X_i)$, 
we show that any other agent does not envy $i$ for more than one item after agent $i$ receives bundle $V^*$. 
Let $M^*_i$ be the maximum matching of $G[V^*]$ for agent $i$.
We first consider agent  $j^* \in A_{l}, l \in [t, \tau]$. Note that replacing $X_i$ by $V^*$ does not change the number of edges in the maximum matching $M_i$ as well as the size of $i$'s bundle $X_i$. Thus, we have
\begin{equation}\nonumber
    u_{j^*}(X_{j^*})=|M_{j^*}| \geq |M_i|=|M^*_i|=\frac{|V^*|}{2}\geq  u_{j^*}(V^*),
\end{equation}
where the last inequality holds because for binary valuations, the valuation of a bundle for one agent is at most half the size of the bundle. Therefore, agent $j^* \in A_{l}, l \in [t, \tau]$ does not envy agent $i$ up to more than one item after $i$ replaces its bundle with $V^*$. 
Next, we consider agent $j^* \in A_{l}, l \in [t-1]$. For the sake of contradiction, assume $u_{j^*}(X_{j^*})<u_{j^*}(V^*)$, which means that there exists at least one edge $e$ such that $w_{j^*}(e)=1$ as well as a bundle $V_{j^*} \subseteq P$ such that $u_{j^*}(V_{j^*})=u_{j^*}(X_{j^*})$. Therefore, in the $l$th round of the {\bf while} loop, a single edge $e$ with weight $1$ is added to agent $j^*$ if it does not break EF1. Otherwise, there exists an agent $j'$ who envies agent $j^*$ before adding edge $e$. In such case, Algorithm \ref{alg:EF1:n04091} will execute the bundle-exchanging procedure in Step \ref{exchanging:from:ef1}-\ref{exchanging:to:ef1} in $l<t$th round of the {\bf while} loop, which is a contradiction with the $t$th round of the {\bf while} loop being executed. We complete the proof of Lemma \ref{lemma;ef1;binary}.
\end{proof}

For any graph $G=(V,E)$ and $n$ different binary valuations $v_i(\cdot)$ on $G$, 
we call a matching $M$ social welfare maximizing if $M$ is a maximum matching on the graph $G'=(V,E')$ where for any $e\in E$ if and only there exists $i$ such that $v_i(e) = 1$.
Let $M^*$ denote the social welfare maximizing matching on the input graph $G$ of Algorithm \ref{alg:EF1:n04091}. 
Let $V_R$ be the set of unallocated items after we move out of the {\bf while} loop in Step \ref{Alg:end_while_1}, and
$M_R$ be the  social welfare maximizing matching on the induced subgraph of $V_R$.
Let $V_L$ be the set of allocated items after Step \ref{Alg:end_while_1} and $M_L$ be the welfare maximizing matching on $V_L$.
Actually, $|M_L|$ is the social welfare that Algorithm \ref{alg:EF1:n04091} produces after Step \ref{Alg:end_while_1}.
Then we have the following.
\begin{lemma}\label{Lemma:L:R:compare}
$|M_L| \geq |M_R|$.
\end{lemma}
\begin{proof}
 We note that when Algorithm \ref{alg:EF1:n04091} moves out of the {\bf while} loop in Step \ref{step:20220411}, any agent $i$ values the unallocated items no more than its own bundle, i.e., $u_i(X_i) \geq u_i(V_R)$. Then we have $|M_L|=\sum\limits_{i=1}^{n}u_i(X_i) \geq \sum\limits_{i=1}^{n}u_i(V_R) \geq |M_R|$, which completes the proof.
\end{proof}
Based on the claims and lemmas presented above, we are ready to prove Theorem \ref{ef1:binary:theorem}
\begin{proof}[Proof of Theorem \ref{ef1:binary:theorem}]
Let  $M_m$ be a welfare maximizing matching on the bipartite graph induced by $V_L$ and $V_R$.
 i.e., finding as many disjoint edges $e_{ij}$ as possible such that $v_i \in V_R$ and $v_j \in V_L$.
Observe that the maximum number of vertices within $V_L$ equals half the number of the edges in the maximum matching $M_L$, i.e., $|V_L|=2|M_L|$. Therefore, the size of $M_m$ is at most $2|M_L|$ (each vertex $v_1 \in V_L$ combined with another vertex $v_2 \in V_R$ to form a matching). Therefore, we have $|M^*| \leq 2|M_L|+|M_R|$. Furthermore, we have
\begin{align*}
    \frac{u(M_L)}{\sw^*} =\frac{|M_L|}{|M^*|}\geq \frac{|M_L|}{2|M_{L}|+|M_R|}  \geq \frac{|M_L|}{3|M_{L}|} \geq \frac{1}{3},
\end{align*}
where the second inequality holds because $|M_L| \geq |M_R|$ proved in Lemma \ref{Lemma:L:R:compare}. 
Since Step \ref{Alg:binary:envy-cycle} can only increase the social welfare, we have proved the social welfare guarantee.

It remains to see the running time of the algorithm. 
In each iteration of the {\bf while} loop in Step \ref{step:20220411}, the utility of exact one agent increases by $1$. Since the maximum possible welfare is bounded by $O(|V|^2)$, the {\bf while} loop will execute for at most $O(|V|^2)$ times. The envy-cycle elimination procedure in Step \ref{Alg:binary:envy-cycle} will execute at most $O(|V|)$ times. Thus, Algorithm \ref{alg:EF1:n04091} runs in $O(|V|^2+|V|)=O(|V|^2)$ time.
\paragraph{Tight Example.}
\begin{figure}[htb]
\begin{center}
\subfigure[Agent $1$'s Metric]{
\includegraphics[width=2.69in]{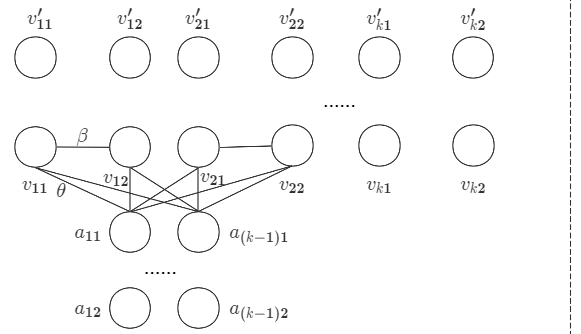}
\label{ef:binary1}
}
\subfigure[Agent $2$'s Metric]{
\includegraphics[width=2.5in]{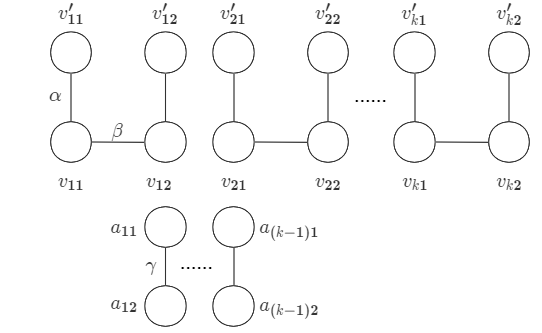}
\label{ef:binary2}
}
\caption{The graph is partitioned among two agents with binary valuations.}
\label{ef:binary3}
\end{center}
\end{figure}
We show that the analysis in Theorem \ref{ef1:binary:theorem} is asymptotically tight.
Consider the example as shown in Figure~\ref{ef:binary3}. Let $k>4$ be a constant. Denote by $\alpha_{ij}, i \in [k], j \in [2]$ the edge between node $v_{ij}$ and node $v'_{ij}$, $\beta_{i}, i \in [k]$ the edge between node $v_{i1}$ and $v_{i2}$, $\gamma_i, i \in [k-1]$ the edge between node $a_{i1}$ and $a_{i2}$. Let $\theta_{i1}, \theta_{i2},\theta_{i3},\theta_{i4},i \in [k]$ be the edge between node $a_{i1}$ and node $v_{11}$, $a_{i1}$ and node $v_{12}$, $a_{i1}$ and node $v_{21}$, $a_{i1}$ and node $v_{22}$, respectively.
Obviously, allocating all the nodes to agent $2$ and allocating nothing to agent $1$ result in the optimal social welfare, i.e.,
\begin{equation}
   \sw^*=u_2(V)=\sum\limits_{i \in [k], j \in [2]}w_2(\alpha_{ij})+\sum\limits_{i \in [k-1]}w_2(\gamma_i)=3k-1. 
\end{equation} 
The corresponding maximum matching $M^*$ contains $2k$ edges $\alpha_{ij}, i \in [k], j \in [2]$ and $k-1$ edges $\gamma_i, i \in [k-1]$. Now, we consider the worst case achieved by Algorithm \ref{alg:EF1:n04091} running on this example, which results in a total utility of $k+3$.
In the first two rounds of the {\bf while} loop in Step \ref{step:20220411}, each agent picks exactly one of the two edges $\beta_1$ and $\beta_2$ (w.l.o.g. agent 1 picks $\beta_1$ and agent 2 picks $\beta_2$). Following that agent $2$ picks all the remaining edges $\beta_i, i \in [3, k]$ and arbitrary two edges $\gamma_i, i \in [k-1]$ (w.l.o.g. $\gamma_1$ and $\gamma_2$). We then move out of the {\bf while} loop since (1) for agent 1, $u_1(P)=0$; (2) for agent 2, $u_2(P) < u_2(X_2)$ and allocating any other edge $\gamma_i, i \in [3, k-1]$ to it will break EF1. Thus, we execute the envy-cycle elimination procedure on the remaining items, i.e.,  allocating all the remaining vertices in $P$ to agent $1$ with the EF1 allocation being completed. For agent 2, the maximum matching in $G[X_2]$ containing edges $\beta_i, i \in [2, k]$, $\gamma_i, i \in [2]$. We thus have $u_2(X_2)=k-1+2=k+1$. For agent 1, the maximum matching in $G[X_1]$ contains edges $a_{11}$ and $a_{12}$. Therefore, $u_1(X_1)=2$. The total social welfare is $u_1(X_1)+u_2(X_2)=k+1+2=k+3$.
Thus
\[
\lim\limits_{k \rightarrow  +\infty}\frac{k+3}{3k-1} =\frac{1}{3},
\]
which completes the proof of the theorem.
\end{proof}
\subsection{Two Heterogeneous Agents}

We then discuss the case of two agents, and show that Algorithm \ref{alg:ef1:2} ensures at least 1/3 fraction of the optimal social welfare.
Intuitively, in Algorithm \ref{alg:ef1:2}, we first check whether there is a single edge $e$ for which some agent $i$ has value at least $1/3\cdot \sw^*(\cI)$.
If so, allocating $e$ to $i$ already ensures $1/3\cdot \sw^*(\cI)$.
Moreover, this partial allocation is EF1 since the removal of one item in $e$ results in no edges, and thus we can use the envy-cycle elimination algorithm to allocate the remaining vertices, which returns an EF1 allocation and can only increase the social welfare.
Otherwise, we compute a social welfare maximizing allocation $(M_1,M_2)$, i.e., $u_1(M_1)+u_2(M_2) = \sw^*(\cI)$.
Without loss of generality, assume $u_1(M_1) \leq u_2(M_2)$.
We temporarily allocate $M_i$ to agent $i$ for $i = 1,2$.
If the allocation is not EF1, since $u_1(M_1) \leq u_2(M_2)$, it can only be the case that agent 1 envies agent 2 but agent 2 does not envy agent 1. 
Then we move items in agent 2's bundle one by one to agent~1.
It can be shown that there must be a time after which the allocation is EF1, and the first time when the allocation becomes EF1, the resulting social welfare is at least $1/3\cdot \sw^*(\cI)$.
Interestingly, despite the simplicity of Algorithm \ref{alg:ef1:2}, we can show that no algorithm that has better than 1/3 approximation. Formally, we have the following theorem.
\begin{theorem}\label{Theorem:ef1:two-agents}
For any instance $\cI$ with two heterogeneous agents, Algorithm \ref{alg:ef1:2} returns an EF1 allocation with social welfare at least $1/3 \cdot \sw^*(\cI)$, which is optimal. 
\end{theorem}

\begin{algorithm}[!htb]
  \caption{EF1 Allocation with tight social welfare guarantee for 2 Heterogeneous Agents
    \label{alg:ef1:2}}
  \begin{algorithmic}[1]
    \REQUIRE Instance $\cI=(G,N,\w)$ with $G=(V, E)$.
  \ENSURE Allocation $\bX=(X_1,  X_2)$.
  \IF{there is $e\in E$ such that $w_i(e) \ge 1/3\cdot \sw^*(\cI)$ for some $i = 1,2$}
  \STATE Assign $e$ to agent $i$ and run envy-cycle elimination  algorithm for the vertices.
  \ELSE 
  \STATE Computing a social welfare maximizing allocation $(M_1, M_2)$. Without loss of generality, assume $u_1(M_1) \leq u_2(M_2)$, and assign $M_i$ to agent $i$ for $i = 1,2$.
  
  \WHILE{agent $1$ envies agent $2$ for more than one item}
          \STATE Reallocate some item $v\in X_2$ to agent $1$, i.e., $X_2 \leftarrow X_2 \setminus \{v\}$ and $X_1 \leftarrow X_1\cup \{v\}$.
  \ENDWHILE
  \ENDIF
       
  \STATE Return the allocation $(X_1, \cdots, X_n)$. 
  \end{algorithmic}
\end{algorithm}
\begin{proof}

Denote $(M_1, M_2)$ as a social welfare maximizing allocation. Consider the following two cases:
\begin{itemize}
    \item Case 1: $\exists e \in E$  such that $w_i(e) \geq \frac{1}{3} \cdot \sw^*, i \in \{1, 2\}$;
    \item Case 2: $\forall e \in E$, $w_i(e) < \frac{1}{3}\cdot \sw^*, i \in \{1, 2\}$.
\end{itemize}
For Case 1, giving edge $e$ to agent $i$ and running the envy-cycle elimination procedure on remaining vertices can find an EF1 allocation, which, at the same time, guarantees the total utility no less than $1/3$ of the maximum possible social welfare. 

Next, we consider Case 2. There are two subcases. 
\begin{itemize}
    \item Subcase 1: $u_i(M_i) \geq \frac{1}{3}\cdot \sw^*$ for all $i \in \{1, 2\}$;
     \item Subcase 2: $\exists i \in \{1, 2\}$ such that $u_i(M_i) < \frac{1}{3}\cdot \sw^*$.
\end{itemize}
For Subcase $1$, if such allocation guarantees EF1, the theorem holds. Otherwise, agent $1$ envies agent $2$ since we assume $u_1(M_1) \leq u_2(M_2)$. We then reallocate the item $v \in X_2$ to agent 1 one by one until such allocation guarantees EF1. The total utility is $u_1(X_1)+u_2(X_2) \geq u_1(M_1) \geq (1/3)\sw^*$. Therefore, we complete the proof for this subcase. 

Consider Subcase $2$. Without loss of generality, assume $u_1(M_1) < (1/3)\sw^*$.
If allocation $(M_1, M_2)$ guarantees EF1, the theorem is proved. Otherwise, by the assumption that $u_1(M_1) \leq u_2(M_2)$, agent $1$ envies agent $2$ more than one item. By $u_1(M_1) < (1/3)\sw^*$, we have $u_2(M_2)>(2/3)\sw^*$. Now, we consider to remove items from agent $2$'s bundle to agent $1$'s bundle. First, we sort the edges within $M_2$ by decreasing order according to their valuation to agent $2$. In each iteration, we pick an edge within agent $2$'s bundle with largest weight and give one endpoint to agent $1$. If the allocation still admits EF1, we give another endpoint to agent $1$ and pick another edge with the largest weight in agent $2$'s remaining bundle. Repeat the above procedure until agent $1$ envies agent $2$ up to exactly one item. When Algorithm \ref{alg:ef1:2} completes, at most one edge $e$ within $M_2$ is destroyed, i.e., one endpoint of $e$ is allocated to agent $1$ and the other endpoint still remains in $X_2$. If $e$ is the edge with the largest weight in $M_2$, we have $u_2(X_2)\geq u_2(M_2 \setminus\{e\})>(1/3)\sw^*$, where the last inequality holds because $u_2(M_2)>(2/3)\sw^*$ and $w_2(e)<(1/3)\sw^*$. We thus complete the proof of the theorem.
Otherwise, we next show that $u_2(X_2) \geq (1/3)u_2(M_2)$. Let $X'_2$ be the set of items given to agent $1$. We have 
\begin{equation}
   w_2(e) \leq u_2(X'_2) \leq u_2(X_2),
  \end{equation}
 where the first inequality holds because at least one edge within $M_2$ with larger weight is allocated to agent $1$ before and the second inequality holds since otherwise agent $2$ will envy agent $1$. We thus derive \begin{equation}
    u_2(X_2)\geq \frac{1}{3}(u_2(X_2)+u_2(X'_2)+w_2(e))\geq \frac{1}{3}u_2(M_2).
\end{equation}
Furthermore
\begin{equation}
\begin{aligned}
  u_1(X_1)+u_2(X_2) &\geq u_1(M_1)+\frac{1}{3}u_2(M_2)\\&\geq \sw^*-u_2(M_2)+\frac{1}{3}u_2(M_2)\\&=\sw^*-\frac{2}{3}u_2(M_2)\geq \frac{1}{3}\sw^*,
\end{aligned}
\end{equation}
where the last inequality holds because  $u_2(M_2) \leq \sw^*$. Since in each iteration, at most one item is removed from agent $2$ to agent $1$, Algorithm \ref{alg:ef1:2} runs in $poly(|V|)$ time. We complete the proof of the theorem.

\paragraph{Tight Example}
\label{subsec:app:2agents:tight}

We next show the approximation of $1/3$ is optimal. Consider the example in Fig. \ref{fig1202203041} and Fig. \ref{fig220203042}. 
It is not hard to verify that the maximum social welfare without fairness constraint is $\sw^*=3$ by allocating all the items to agent 1.
However, for any allocation where agent 1 has utility no smaller than 2, the allocation is not EF1 to agent 2 since agent 2 always has utility 0 in such allocations. 
Therefore, the maximum social welfare generated by EF1 allocations is no greater than $1+2\epsilon$.
Thus
\begin{equation}
    \lim_{x \to 0}\frac{1+2\epsilon}{3}=\frac{1}{3},
\end{equation}
which means the approximation ratio of $1/3$ is optimal.
\end{proof}
\begin{figure*}
\centering
\subfigure[Agent $1$'s weight for the graph]{
\includegraphics[width=1.6in]{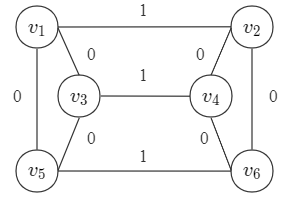}
\label{fig1202203041}
}
\hspace{.8in}
\subfigure[Agent $2$'s weight for the graph]{
\includegraphics[width=1.8in]{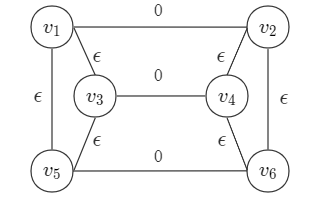}
\label{fig220203042}
}
\vspace{-0.1in}
\caption{An example where any EF1 allocation guarantees at most $(1/3+2\epsilon)$ of the maximum social welfare.}
\label{fig202108250}
\end{figure*}

\subsection{Homogeneous Agents}\label{envyfree}
\begin{algorithm}[htb]
  \caption{Computing EF1 Allocations with High Social Welfare for $n$ Homogeneous Agents
    \label{alg:EF1:n}}
  \begin{algorithmic}[1]
  \REQUIRE Instance $\cI=(G,N)$ with $G=(V, E; w)$.
  \ENSURE Allocation $\bX=(X_1, \cdots, X_n)$.
    \STATE Find a maximum matching $M^*$ in $G$.
   Denote by $V'$ the set of unmatched vertices by $M^*$.
   \STATE Find the greedy partition $(M_1,\cdots,M_n)$ of edges in $M^*$ such that $w(M_1) \ge \cdots \ge w(M_n)$.
   \STATE Set $X_i = V(M_i)$ for $i =1, \cdots, n$.
   \IF{$|M^*| \le n$}
        \STATE  Let $X_n=V(M_n)\cup V'$.
        \STATE \label{step:ef1:n:1} Return $(X_1,\cdots,X_n)$.
    \ENDIF
    \STATE Construct the EF1-graph $\cG=(N, \cE)$ based on  $(X_1,\cdots,X_n)$.
    \STATE Set $Q$ be the agents with positive in-degree. 
    \FOR{$i \in Q$} \label{step:EF1:n:for1}
        \STATE Let $e_i=(v_{i1}, v_{i2})$ be the last edge added to $M_i$ in the greedy-partition procedure. \label{step:EF1:n:for1:e}
        \STATE $X_i = X_i \setminus \{v_{i1}\}$ and $V' = V' \cup \{v_{i1}\}$.
     \ENDFOR
    \FOR{$v\in V'$} \label{step:EF1:n:for2}
        \STATE Let $i=\arg\min_{i\in N} u(X_i)$.
        \STATE Set $X_i = X_i \cup \{v\}$.
    \ENDFOR
    \STATE \label{step:ef1:n:2} Return $(X_1,\cdots,X_n)$.
   
  \end{algorithmic}
\end{algorithm}
\begin{theorem}\label{thm:ef1:n}
For any homogeneous instance $\cI$, 
Algorithm~\ref{alg:EF1:n} returns an EF1 allocation with social welfare at least $(2/3+2/(9n-3))\cdot \sw^*(\cI)$  in polynomial time.
\end{theorem}
In the following, we  first briefly discuss the idea of Algorithm~\ref{alg:EF1:n}.
We introduce the  {\em EF1-graph}, inspired by the envy-graph introduced in \citep{DBLP:conf/sigecom/LiptonMMS04}.
Given a (partial) allocation $(X_1,\cdots,X_n)$, we construct the corresponding EF1-graph $\cG = (N, \cE)$, where the nodes are agents (and thus are used interchangeably) and there is a directed edge from $i$ to $j$ if $i$ envies $j$ (or $X_j$) for more than one item, 
\[
u_i(X_i) < u_i(X_j \setminus \{v\}) \text{ for every $v \in X_j$}.
\]
When the agents have identical utility functions, we have the following simple observation. 
\begin{observation}\label{ob:ef1:graph}
The EF1-graph is acyclic;
The in-degree of the agent with smallest utility is zero.
\end{observation}
Similar with Algorithm \ref{alg20210091}, in Algorithm \ref{alg:EF1:n}, we first compute a maximum weighted matching $M^*$ and let the corresponding unmatched vertices be $V'$.
If $|M^*| \le n$, by allocating each edge in $M^*$ to a different agent and $V'$ to one agent who has the smallest utility is EF1, since by removing a vertex from an edge, the remaining subgraph does not have edges any more. 
If $|M^*| > n$, we find a greedy-partition $\Gamma(M^*)=(M_1,\cdots,M_n)$ of $M^*$ such that $w(M_1) \ge \cdots \ge w(M_n)$.
However, by simply assigning $X_i = V(M_i)$ for every $i$, it may not be EF1, which is illustrated in the appendix.

To overcome this difficulty, we utilize the EF1-graph $\cG = (N,\cE)$ on the partial allocation $(V(M_1),  \cdots, V(M_n))$.
Let $Q\subseteq N$ be the set of agents who have positive in-degree, i.e., are envied by some agent for more than one item. 
By Observation \ref{ob:ef1:graph}, if $\cG$ is nonempty, $Q\neq \emptyset$ and $n \notin Q$.
Moreover, since $M_n$ has the smallest weight in the greedy partition $\Gamma(M^*)$, $n$ has an edge to every agent in $Q$.
We first consider the partial allocation after the {\bf for} loop in Step \ref{step:EF1:n:for1}, which is denoted by $Y=(Y_1,\cdots,Y_n)$.
We can prove that $Y$ is EF1, and moreover, it ensures the desired social welfare guarantee. 
Finally, the remaining steps preserve the EF1ness and can only increase the social welfare of the allocation.
Before proving Theorem \ref{thm:ef1:n}, we first show several technical lemmas.

\begin{lemma}
\label{lem:ef1:n:ef1}
$Y$ is EF1.
\end{lemma}
\begin{proof}
If $Q=\emptyset$, by definition, the allocation is already EF1.
In the following, assume $Q\neq\emptyset$.
Note that only the agents in $i\in Q$ has one vertex removed from $V(M_i)$ and for any $i\notin Q$, $Y_i = V(M_i)$.
Particularly, $Y_n=V(M_n)$.

Fix any $i \in Q$.
Let $(v_{i1},v_{i2})$ be the edge selected in Step \ref{step:EF1:n:for1:e}, i.e., the edge with the smallest weight in $M_i$.
By the definition of greedy partition, 
\begin{align}\label{eq:ef1:n:pef1}
    u(V(M_{n})) \ge u(V(M_{i}) \setminus \{v_{i1},v_{i2}\}).
\end{align}
We have the following claims. 

\begin{claim}
\label{claim:ef1:n:ef1:n}
Agent $n$ does not envy any agent $i \in N$ for more than one item in the partial allocation $Y$.
\end{claim}
\begin{proof}
The claim is straightforward if $i\notin Q$ since there is no edge between $n$ and $i$.
If $i\in Q$, then $Y_i = V(M_{i}) \setminus \{v_{i1}\}$ and by Inequality \eqref{eq:ef1:n:pef1},
\[
u(V(M_{n})) \ge u(V(M_{i}) \setminus \{v_{i1},v_{i2}\}) = u(Y_i\setminus\{v_{i2}\}),
\]
implying $n$ does not envy $Y_i$ for more than one item.
\end{proof}

\begin{claim}
\label{claim:ef1:n:ef1}
No agent $i \in N$ envies agent $n$ in $Y$.
\end{claim}
\begin{proof}
If $i\notin Q$, the bundles of agent $i$ and $n$ do not change in the {\bf for} loop in Step \ref{step:EF1:n:for1}. Since $M_n$ has the smallest weight in the greedy partition of $M^*$, we have
\[
u(Y_i)=u(V(M_i)) \ge u(V(M_n)) = u(Y_n).
\]
If $i\in Q$, since there is an edge from $n$ to $i$, we have
\[
u(Y_i) \ge \min_{v\in V(M_i)} u(V(M_i)\setminus\{v\}) > u(V(M_n)),
\]
which means $i$ does not envy $n$. 
\end{proof}
Combining Claims \ref{claim:ef1:n:ef1:n} and \ref{claim:ef1:n:ef1}, we have that for any two agents $i$ and $j$,
\[
u(Y_i) \ge u(Y_n) \ge u(Y_j \setminus\{v\}) \text{ for some $v\in Y_j$},
\]
which means $i$ does not envy $j$ for more than one item. 
This completes the proof of Lemma \ref{lem:ef1:n:ef1}.
\end{proof}




To prove the approximation ratio of Algorithm \ref{alg:EF1:n}, we need the following lemma.
\begin{lemma}\label{lem:ef1:n:mig3}
$|M_i| \ge 3$ for all $i\in Q$.
\end{lemma}
\begin{proof}
If the in-degree of agent $i$ is non-zero, then agent $n$ must envy $i$ for more than one item, and 
\begin{align}\label{eq:ef1:mig3}
    u(Y_n) < u(V(M_i)\setminus\{v\}) \text{ for any $v\in V(M_i)$}.
\end{align}
First, it is easy to see that $|M_i| \neq 1$ since the removal of any node $v$ makes the remaining utility be 0 and thus Equation \eqref{eq:ef1:mig3} does not hold.

Next we show  $|M_i| \neq 2$.
For the sake of contradiction, assume $M_i=\{e,e'\}$ with $e=(v_1,v_2)$ and $e'=(v'_1,v'_2)$.
Without loss of generality, we further assume $w(e) \ge w(e')$.
Then it must be that $w(M_n) \ge w(e)$, otherwise $e'$ cannot be added to $M_i$.
Note that since $M^*$ is a maximum weighted matching in $G$, $\{e,e'\}$ must be a maximum weighted matching in $G[M_i]$.
If there exist edges in $G[M_i]$ whose weights are greater than $w(e)$, these edges must be adjacent to the same node, denoted by $\bar{v}$; otherwise they can form another matching with weight greater than $w(M_i)$. 
Thus by removing $\bar{v}$ from $G[M_i]$, the maximum matching in the remaining graph contains at most one edge, and all the remaining edges have weight at most $w(e)$, which means the maximum matching in $G[V(M_i)\setminus\{\bar{v}\}]$ brings utility no larger than $w(e)$. 
Therefore,
\[
u(V(M_i)\setminus\{\bar{v}\}) \le w(e) \le w(M_n),
\]
which is a contradiction with Equation \eqref{eq:ef1:mig3}.
Combining the above two cases, we have $|M_i| \ge 3$.
\end{proof}



Based on the claims and lemmas presented above, we present the proof of Theorem \ref{thm:ef1:n} below.

\begin{proof}[Proof of Theorem \ref{thm:ef1:n}]
Let $(X_1,\cdots,X_n)$ be the allocation returned by Algorithm \ref{alg:EF1:n}.
If the allocation is from Step \ref{step:ef1:n:1}, then it must be EF1.
This is because $X_n$ has the smallest value and thus nobody envies $n$ and each of $X_i$ with $1\le i\le n-1$ contains only two nodes which means the removal of one of them brings utility 0 to any agent.
It also achieves the optimal social welfare since all edges in $M^*$ are allocated to some agents. 

Next we consider the case when the allocation is obtained from Step \ref{step:ef1:n:2}. 
By Lemma \ref{lem:ef1:n:ef1}, after the {\bf for} loop in Step \ref{step:EF1:n:for1}, the partial allocation is EF1. 
To show the final allocation to be EF1, it suffices to show that the {\bf for} loop in Step \ref{step:EF1:n:for2} preserves EF1. 
This is true as in each round, only the bundle with the smallest value can be allocated one more item whose removal makes it smallest again. 


Finally, we consider the social welfare loss. 
For each agent $i\in Q$, we observe that at most one node will be removed from $V(M_i)$ in the {\bf for} loop in Step \ref{step:EF1:n:for1} and the {\bf for} loop in Step \ref{step:EF1:n:for2} can only increase $i$'s utility. 
Since the removed node $v_{i1}$ is from the edge with the smallest weight in $M_i$, by Lemma \ref{lem:ef1:n:mig3}, we have
\[
u(X_i) \ge \frac{2}{3} \cdot u(V(M_i)) \text{ for all $i\in N\setminus\{n\}$},
\]
Moreover, for agent $n$ and any $i\neq n$,
\begin{align}\label{eq:ef1:n:sw}
    u(X_n) \ge u(M_i\setminus \{(v_{i1},v_{i2})\})\ge \frac{2}{3} \cdot u(V(M_i)).
\end{align}
Therefore
\begin{align*}
    \frac{\sum_{i\in N}u(X_i)}{\sw^*} 
    &\ge \dfrac{\sum_{i\in N\setminus\{n\}} \frac{2}{3}\cdot u(V(M_i)) + u(V(M_n))}{\sum_{i\in N} u(V(M_i))} \\
    &=\frac{2}{3} + \frac{1}{3} \cdot \dfrac{u(V(M_n))}{\sum_{i\in N} u(V(M_i))}\\
    &\ge\frac{2}{3} + \frac{1}{3} \cdot \dfrac{u(V(M_n))}{ (\frac{3}{2} (n-1)+1)\cdot u(V(M_n))}\\
    &=\frac{2}{3} + \frac{2}{9n-3},
\end{align*}
where the second inequality is because of Inequality \eqref{eq:ef1:n:sw} and we complete the proof of Theorem \ref{thm:ef1:n}.
\end{proof}

\paragraph{Tight Example.}
We show that the analysis in Theorem \ref{thm:ef1:n} is asymptotically tight.
Consider the example in Figure \ref{figure20220112}, where $2^-$ means $2-\epsilon^2$ and $4^+$ means $4+3(n-1)\epsilon^2$.
Let $\epsilon>0$ be a sufficiently small number, say $1/n^2$.
The maximum matching $M^*$ contains all the bold edges and $\sw^*=w(M^*) = 6(n-1)+4$.
By Algorithm~\ref{alg:EF1:n}, the greedy-partition of $M^*$ is $(M_1,\cdots,M_n)$ as shown in Figure \ref{figure20220112}.
However, it is not EF1: for $1\le i\le n-1$, by removing any vertex from $M_i$, the maximum matching in the remaining graph has weight at least $2+\epsilon+2^- > 4+ 3(n-1)\epsilon^2 = w(M_n)$.
After the {\bf for} loop in Step \ref{step:EF1:n:for1} in Algorithm \ref{alg:EF1:n}, for $1\le i\le n-1$, one vertex in each $M_i$ is removed and is reallocated to $M_n$ in the {\bf for} loop in Step \ref{step:EF1:n:for2}.
Thus the remaining social welfare is at most 
\[
2 \cdot (2+\epsilon)\cdot  (n-1) + 4 \to \frac{2}{3}\cdot \sw^*.
\]

\begin{figure}
\centering
\includegraphics[width=2.8in]{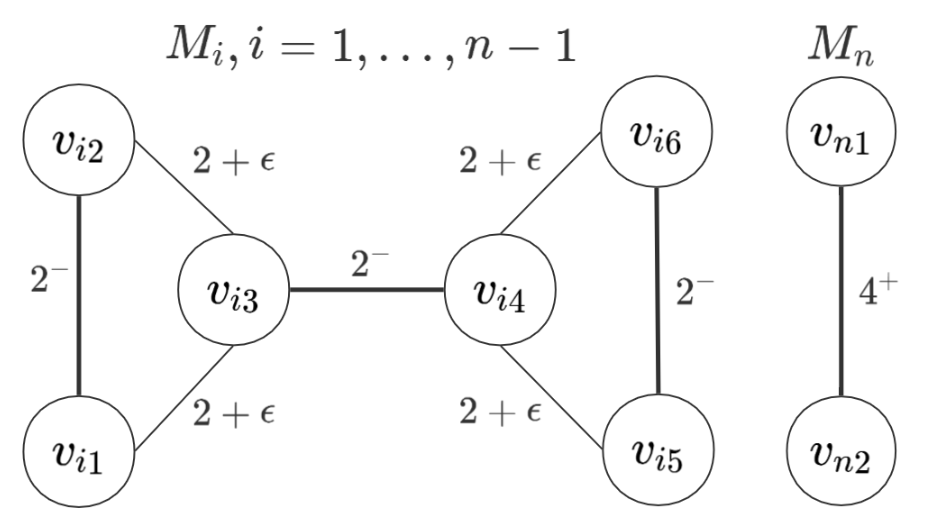}
\caption{A graph contains $n$ connected components where the first $n-1$ components are identical as shown by $M_i,i=1,\cdots,n-1$, and the last component is a single edge as shown by $M_n$. }\label{figure20220112}
\end{figure}

\paragraph{Remark.}
By Theorem \ref{thm:ef1:n}, if $n=2$, the approximation ratio is 4/5 and when $n\to \infty$ the approximation ratio is 2/3.
Unfortunately, we were not able to prove an upper bound where the optimal social welfare cannot be achieved by any EF1 allocation.
We conjecture that there is always an EF1 allocation that achieves the optimal social welfare $\sw^*$.
\section{Conclusion and Future Directions}\label{sec:Extension}

In this work, we study the fair (and efficient) allocation of graphical resources when the agents' utilities are determined by the weights of the maximum matchings in the obtained subgraphs. 
We provide a string of algorithmic results regarding MMS and EF1, but also leave some problems open. 
For example, regarding MMS, we can further improve the approximation ratio when the agents are homogeneous, and prove inapproximability results;
regarding EF1, the approximation ratio to the optimal social welfare for binary weight functions and homogeneous agents can be potentially improved. 

Our work also uncovers some other interesting future directions. 
First, regarding MMS, although we show that there is no bounded multiplicative approximation, it may admit good additive or bi-factor approximations.
Second, we only focus on the matching-induced utilities in this work, and it is intriguing to consider other combinatorial structures such as independent set, network flow and more.
Third, we can extend the framework to the fair allocation of graphical chores when the agents have costs to complete the items, and the asymmetric situation when the agents have possibly different entitlements to the system.

{
\bibliographystyle{named}
\bibliography{name1}

\begin{thebibliography}{}

\bibitem[\protect\citeauthoryear{Amanatidis \bgroup \em et al.\egroup
  }{2022}]{DBLP:journals/corr/abs-2202-07551}
Georgios Amanatidis, Georgios Birmpas, Aris Filos{-}Ratsikas, and Alexandros~A.
  Voudouris.
\newblock Fair division of indivisible goods: {A} survey.
\newblock {\em CoRR}, abs/2202.07551, 2022.

\bibitem[\protect\citeauthoryear{Aziz \bgroup \em et al.\egroup
  }{2022}]{DBLP:journals/corr/abs-2202-08713}
Haris Aziz, Bo~Li, Herv{\'{e}} Moulin, and Xiaowei Wu.
\newblock Algorithmic fair allocation of indivisible items: {A} survey and new
  questions.
\newblock {\em CoRR}, abs/2202.08713, 2022.

\bibitem[\protect\citeauthoryear{Bansal and
  Sviridenko}{2006}]{DBLP:conf/stoc/BansalS06}
Nikhil Bansal and Maxim Sviridenko.
\newblock The santa claus problem.
\newblock In {\em {STOC}}, pages 31--40, 2006.

\bibitem[\protect\citeauthoryear{Barketau \bgroup \em et al.\egroup
  }{2015}]{DBLP:journals/dam/BarketauPS15}
Maksim Barketau, Erwin Pesch, and Yakov~M. Shafransky.
\newblock Minimizing maximum weight of subsets of a maximum matching in a
  bipartite graph.
\newblock {\em Discret. Appl. Math.}, 196:4--19, 2015.

\bibitem[\protect\citeauthoryear{Barman and
  Krishnamurthy}{2020}]{DBLP:journals/teco/BarmanK20}
Siddharth Barman and Sanath~Kumar Krishnamurthy.
\newblock Approximation algorithms for maximin fair division.
\newblock {\em {ACM} Trans. Economics and Comput.}, 8(1):5:1--5:28, 2020.

\bibitem[\protect\citeauthoryear{Bei \bgroup \em et al.\egroup
  }{2021}]{DBLP:journals/mst/BeiLMS21}
Xiaohui Bei, Xinhang Lu, Pasin Manurangsi, and Warut Suksompong.
\newblock The price of fairness for indivisible goods.
\newblock {\em Theory Comput. Syst.}, 65(7):1069--1093, 2021.

\bibitem[\protect\citeauthoryear{Bil{\`{o}} \bgroup \em et al.\egroup
  }{2019}]{DBLP:conf/innovations/BiloCFIMPVZ19}
Vittorio Bil{\`{o}}, Ioannis Caragiannis, Michele Flammini, Ayumi Igarashi,
  Gianpiero Monaco, Dominik Peters, Cosimo Vinci, and William~S. Zwicker.
\newblock Almost envy-free allocations with connected bundles.
\newblock In {\em {ITCS}}, volume 124 of {\em LIPIcs}, pages 14:1--14:21.
  Schloss Dagstuhl - Leibniz-Zentrum f{\"{u}}r Informatik, 2019.

\bibitem[\protect\citeauthoryear{Bouveret \bgroup \em et al.\egroup
  }{2017}]{DBLP:conf/ijcai/BouveretCEIP17}
Sylvain Bouveret, Katar{\'{\i}}na Cechl{\'{a}}rov{\'{a}}, Edith Elkind, Ayumi
  Igarashi, and Dominik Peters.
\newblock Fair division of a graph.
\newblock In {\em {IJCAI}}, pages 135--141. ijcai.org, 2017.

\bibitem[\protect\citeauthoryear{Budish}{2011}]{budish2011combinatorial}
Eric Budish.
\newblock The combinatorial assignment problem: Approximate competitive
  equilibrium from equal incomes.
\newblock {\em Journal of Political Economy}, 119(6):1061--1103, 2011.

\bibitem[\protect\citeauthoryear{Bulu{\c{c}} \bgroup \em et al.\egroup
  }{2016}]{DBLP:series/lncs/BulucMSS016}
Aydin Bulu{\c{c}}, Henning Meyerhenke, Ilya Safro, Peter Sanders, and Christian
  Schulz.
\newblock Recent advances in graph partitioning.
\newblock In {\em Algorithm Engineering}, volume 9220 of {\em Lecture Notes in
  Computer Science}, pages 117--158. 2016.

\bibitem[\protect\citeauthoryear{Crouch and Mazur}{2001}]{crouch2001peer}
Catherine~H Crouch and Eric Mazur.
\newblock Peer instruction: Ten years of experience and results.
\newblock {\em American journal of physics}, 69(9):970--977, 2001.

\bibitem[\protect\citeauthoryear{Farbstein and
  Levin}{2015}]{DBLP:journals/disopt/FarbsteinL15}
Boaz Farbstein and Asaf Levin.
\newblock Min-max cover of a graph with a small number of parts.
\newblock {\em Discret. Optim.}, 16:51--61, 2015.

\bibitem[\protect\citeauthoryear{Flanigan \bgroup \em et al.\egroup
  }{2021}]{flanigan2021fair}
Bailey Flanigan, Paul G{\"o}lz, Anupam Gupta, Brett Hennig, and Ariel~D
  Procaccia.
\newblock Fair algorithms for selecting citizens’ assemblies.
\newblock {\em Nature}, pages 1--5, 2021.

\bibitem[\protect\citeauthoryear{Foley}{1967}]{foley1967resource}
D.~K. Foley.
\newblock Resource {{Allocation}} and the {{Public Sector}}.
\newblock {\em Yale Econ. Essays}, 7, 1967.

\bibitem[\protect\citeauthoryear{Garg and Taki}{2021}]{garg2019improved}
Jugal Garg and Setareh Taki.
\newblock An improved approximation algorithm for maximin shares.
\newblock {\em Artificial Intelligence}, 300, 2021.

\bibitem[\protect\citeauthoryear{Ghodsi \bgroup \em et al.\egroup
  }{2018}]{DBLP:conf/sigecom/GhodsiHSSY18}
Mohammad Ghodsi, Mohammad~Taghi Hajiaghayi, Masoud Seddighin, Saeed Seddighin,
  and Hadi Yami.
\newblock Fair allocation of indivisible goods: Improvements and
  generalizations.
\newblock In {\em {EC}}, pages 539--556, 2018.

\bibitem[\protect\citeauthoryear{Goldman and
  Procaccia}{2014}]{DBLP:journals/sigecom/GoldmanP14}
Jonathan~R. Goldman and Ariel~D. Procaccia.
\newblock Spliddit: unleashing fair division algorithms.
\newblock {\em SIGecom Exch.}, 13(2):41--46, 2014.

\bibitem[\protect\citeauthoryear{Igarashi and
  Peters}{2019}]{DBLP:conf/aaai/IgarashiP19}
Ayumi Igarashi and Dominik Peters.
\newblock Pareto-optimal allocation of indivisible goods with connectivity
  constraints.
\newblock In {\em {AAAI}}, pages 2045--2052. {AAAI} Press, 2019.

\bibitem[\protect\citeauthoryear{Khani and
  Salavatipour}{2014}]{DBLP:journals/algorithmica/KhaniS14}
M.~Reza Khani and Mohammad~R. Salavatipour.
\newblock Improved approximation algorithms for the min-max tree cover and
  bounded tree cover problems.
\newblock {\em Algorithmica}, 69(2):443--460, 2014.

\bibitem[\protect\citeauthoryear{Ko{\c{c}} \bgroup \em et al.\egroup
  }{2016}]{DBLP:journals/eor/KocBJL16a}
{\c{C}}agri Ko{\c{c}}, Tolga Bektas, Ola Jabali, and Gilbert Laporte.
\newblock Thirty years of heterogeneous vehicle routing.
\newblock {\em Eur. J. Oper. Res.}, 249(1):1--21, 2016.

\bibitem[\protect\citeauthoryear{Kress \bgroup \em et al.\egroup
  }{2015}]{DBLP:journals/eor/KressMP15}
Dominik Kress, Sebastian Meiswinkel, and Erwin Pesch.
\newblock The partitioning min-max weighted matching problem.
\newblock {\em Eur. J. Oper. Res.}, 247(3):745--754, 2015.

\bibitem[\protect\citeauthoryear{Kurokawa \bgroup \em et al.\egroup
  }{2018}]{kurokawa2018fair}
D.~Kurokawa, A.~Procaccia, and J.~Wang.
\newblock Fair enough: Guaranteeing approximate maximin shares.
\newblock {\em Journal of the ACM}, 65(2):8, 2018.

\bibitem[\protect\citeauthoryear{Lipton \bgroup \em et al.\egroup
  }{2004}]{DBLP:conf/sigecom/LiptonMMS04}
Richard~J. Lipton, Evangelos Markakis, Elchanan Mossel, and Amin Saberi.
\newblock On approximately fair allocations of indivisible goods.
\newblock In {\em {EC}}, pages 125--131. {ACM}, 2004.

\bibitem[\protect\citeauthoryear{Lov{\'a}sz and
  Plummer}{2009}]{lovasz2009matching}
L{\'a}szl{\'o} Lov{\'a}sz and Michael~D Plummer.
\newblock {\em Matching theory}, volume 367.
\newblock American Mathematical Soc., 2009.

\bibitem[\protect\citeauthoryear{Miyazawa \bgroup \em et al.\egroup
  }{2021}]{DBLP:journals/eor/MiyazawaMOW21}
Fl{\'{a}}vio~Keidi Miyazawa, Phablo F.~S. Moura, Matheus~J. Ota, and Yoshiko
  Wakabayashi.
\newblock Partitioning a graph into balanced connected classes: Formulations,
  separation and experiments.
\newblock {\em Eur. J. Oper. Res.}, 293(3):826--836, 2021.

\bibitem[\protect\citeauthoryear{Moulin}{2003}]{DBLP:books/daglib/0017734}
Herv{\'{e}} Moulin.
\newblock {\em Fair division and collective welfare}.
\newblock {MIT} Press, 2003.

\bibitem[\protect\citeauthoryear{Rathinam \bgroup \em et al.\egroup
  }{2020}]{DBLP:conf/swat/Rathinam0BS20}
Sivakumar Rathinam, R.~Ravi, J.~Bae, and Kaarthik Sundar.
\newblock Primal-dual 2-approximation algorithm for the monotonic multiple
  depot heterogeneous traveling salesman problem.
\newblock In {\em {SWAT}}, volume 162 of {\em LIPIcs}, pages 33:1--33:13, 2020.

\bibitem[\protect\citeauthoryear{Suksompong}{2019}]{DBLP:journals/dam/Suksompong19}
Warut Suksompong.
\newblock Fairly allocating contiguous blocks of indivisible items.
\newblock {\em Discret. Appl. Math.}, 260:227--236, 2019.

\bibitem[\protect\citeauthoryear{Suksompong}{2021}]{DBLP:journals/sigecom/Suksompong21}
Warut Suksompong.
\newblock Constraints in fair division.
\newblock {\em SIGecom Exch.}, 19(2):46--61, 2021.

\bibitem[\protect\citeauthoryear{Traub and
  Tr{\"{o}}bst}{2020}]{DBLP:conf/ipco/TraubT20}
Vera Traub and Thorben Tr{\"{o}}bst.
\newblock A fast {(2} + 2/7)-approximation algorithm for capacitated cycle
  covering.
\newblock In {\em {IPCO}}, pages 391--404. Springer, 2020.

\bibitem[\protect\citeauthoryear{Yaman}{2006}]{DBLP:journals/mp/Yaman06}
Hande Yaman.
\newblock Formulations and valid inequalities for the heterogeneous vehicle
  routing problem.
\newblock {\em Math. Program.}, 106(2):365--390, 2006.

\end{thebibliography}
}

\newpage
\end{document}